%% file: anonymous-submission-latex-2025.tex
\documentclass[letterpaper]{article} 
\usepackage{aaai25}  
\usepackage{times}  
\usepackage{helvet}  
\usepackage{courier}  
\usepackage[hyphens]{url}  
\usepackage{graphicx} 
\usepackage{natbib}  
\usepackage{caption} 
\usepackage{algorithm}
\usepackage[noend]{algorithmic}

\input{macros}

\usepackage{newfloat}
\usepackage{listings}
\DeclareCaptionStyle{ruled}{labelfont=normalfont,labelsep=colon,strut=off} 
\floatstyle{ruled}
\newfloat{listing}{tb}{lst}{}
\floatname{listing}{Listing}
\title{Hyperparametric Robust and Dynamic Influence Maximization}
\author {
    Arkaprava Saha\textsuperscript{\rm 1,\rm 2},
    Bogdan Cautis\textsuperscript{\rm 3},
    Xiaokui Xiao\textsuperscript{\rm 4},
    Laks V.S. Lakshmanan\textsuperscript{\rm 5}
}
\affiliations {
    \textsuperscript{\rm 1}DesCartes Program, CNRS@CREATE, Singapore\\
    \textsuperscript{\rm 2}Laboratoire d'Informatique de Grenoble, Universit{\'e} Grenoble Alpes, Grenoble, Auvergne-Rh{\^o}ne-Alpes, France\\
    \textsuperscript{\rm 3}Laboratoire Interdiscplinaire des Sciences du Num{\'e}rique, Universit{\'e} Paris Saclay, Paris, {\^I}le-de-France, France\\
    \textsuperscript{\rm 4}School of Computing, National University of Singapore, Singapore\\
    \textsuperscript{\rm 5}Department of Computer Science, University of British Columbia, Vancouver, British Columbia, Canada\\
    arkaprava.saha@\{cnrsatcreate.sg, univ-grenoble-alpes.fr\}, bogdan.cautis@universite-paris-saclay.fr, xkxiao@nus.edu.sg, laks@cs.ubc.ca
}

\usepackage{bibentry}

\begin{document}

\maketitle

\begin{abstract}
We study the problem of robust influence maximization in dynamic diffusion networks. In line with recent works, we consider the scenario where the network can undergo insertion and removal of nodes and edges, in discrete time steps, and the influence weights are determined by the features of the corresponding nodes and a global hyperparameter. Given this, our goal is to find, at every time step, the seed set maximizing the worst-case influence spread across all possible values of the hyperparameter. 
We propose an approximate solution using multiplicative weight updates and a greedy algorithm, with provable quality guarantees. 
Our experiments validate the effectiveness and efficiency of the proposed methods.
\end{abstract}

%

\input{introduction}
\input{related}
\input{preliminaries}
\input{solution}
\input{experiments}

\input{conclusion}
\section*{Acknowledgments}
Arkaprava Saha's research was supported by the National Research Foundation, Prime Minister’s Office, Singapore under its Campus for Research Excellence and Technological Enterprise (CREATE) programme. Laks V.S. Lakshmanan's research was supported in part by a grant from the Natural Sciences and Engineering Research Council of Canada (Grant Number RGPIN-2020-05408). The computational work for this article was performed fully on resources of the National Supercomputing Centre, Singapore (https://www.nscc.sg).

\bibliography{aaai25}



\end{document}

%% file: macros.tex
\usepackage{graphicx}
\usepackage{booktabs}
\usepackage{amsmath}
\usepackage{amsthm}
\usepackage{amsfonts}
\usepackage{epstopdf}
\usepackage{graphicx}
\usepackage{subcaption}
\usepackage{url}
\usepackage[noend]{algorithmic}
\usepackage{algorithm}
\usepackage{multirow}
\usepackage{xspace}
\usepackage{color}
\usepackage{cleveref}

\usepackage{pifont}

\newcommand\eat[1]{}

\crefrangeformat{subfigure}{(#3\onlyletter{#1}#4--#5\onlyletter{#2}#6)}
\ExplSyntaxOn
\NewDocumentCommand{\onlyletter}{m}
 {
  \tl_set:Nx \l_tmpa_tl { #1 }
  \tl_item:Nn \l_tmpa_tl { -1 }
 }
\ExplSyntaxOff

\newtheorem{theorem}{Theorem}

\newtheorem{corollary}{Corollary}
\newtheorem{lemma}{Lemma}

\newtheorem{problem}{Problem}

\newcommand{\spara}[1]{\smallskip\noindent{\bf #1}}

\newcommand{\squishlist}{
 \begin{list}{$\bullet$}
  {  \setlength{\itemsep}{0pt}
     \setlength{\parsep}{3pt}
     \setlength{\topsep}{3pt}
     \setlength{\partopsep}{0pt}
     \setlength{\leftmargin}{2em}
     \setlength{\labelwidth}{1.5em}
     \setlength{\labelsep}{0.5em}
} }
\newcommand{\squishlisttight}{
 \begin{list}{$\bullet$}
  { \setlength{\itemsep}{0pt}
    \setlength{\parsep}{0pt}
    \setlength{\topsep}{0pt}
    \setlength{\partopsep}{0pt}
    \setlength{\leftmargin}{2em}
    \setlength{\labelwidth}{1.5em}
    \setlength{\labelsep}{0.5em}
} }

\newcommand{\squishdesc}{
 \begin{list}{}
  {  \setlength{\itemsep}{0pt}
     \setlength{\parsep}{3pt}
     \setlength{\topsep}{3pt}
     \setlength{\partopsep}{0pt}
     \setlength{\leftmargin}{1em}
     \setlength{\labelwidth}{1.5em}
     \setlength{\labelsep}{0.5em}
} }

\newcommand{\squishend}{
  \end{list}
}

\newcommand{\NP}{\ensuremath{\mathbf{NP}}\xspace}

\newcounter{ccc}

\DeclareMathOperator*{\argmax}{arg\,max}
\newcommand{\bigO}{\mathcal{O}}

\usepackage{soul} 
\newcommand\redout{\bgroup\markoverwith
{\textcolor{red}{\rule[.5ex]{2pt}{2pt}}}\ULon}

\newcommand{\algoname}{RIME}

%% file: introduction.tex
\section{Introduction}
\label{sec:intro}
Motivated by applications in viral marketing \cite{DR01},  infection detection \cite{leskovec2007cost}, or misinformation mitigation \cite{simpson2022misinformation}, information and influence diffusion over social networks have been extensively studied  \cite{GionisTT13,
TSX15,li2018influence}. 
The interest in and the applicability of studies of information propagation goes beyond  social media, as diffusion phenomena and algorithms for 
optimizing over viral  mechanisms may fit in diverse domains such as  Internet virus spreading \cite{PhysRevLett.86.3200}, 
biological systems \cite{https://doi.org/10.1046/j.1461-0248.2000.00130.x}, failures in power grids \cite{Kinney_2005}, or synchronization in ad-hoc communication networks \cite{6125993}.

There are many ways to model information propagation in diffusion networks. A popular model is Independent Cascades (IC)  \cite{KKT03}, which is \emph{markovian}  (memoryless), in that the \emph{activation probability} of a node only depends on the current state of the system, and \emph{local}, i.e., the activation of one node can only be caused by its neighbors. Given such a stochastic diffusion model, the classic influence maximization (IM) problem aims to find the top-$k$  seed nodes maximizing the expected number of influenced nodes in the network. 
The IC model has been widely and successfully used in many IM settings, effectively capturing the key aspects of diffusion phenomena. 


However, the most widely used diffusion models, including IC, have two major shortcomings that limit their practical applicability. Firstly, they assume that the influence weight (or probability) between each pair of nodes is \textit{known accurately beforehand}. When this assumption does not hold, i.e., the model is uncertain, existing IM algorithms can easily fail: small perturbations in these weights can lead to drastic changes in the solution quality \cite{GBL11, adiga2014sensitivity}, and even approximating the stability of a solution against  small perturbations is intractable \cite{he2014stability}.
However, model uncertainty is unavoidable in practice, since  diffusion probabilities need to be learned from available diffusion cascades. To address this, in recent years \emph{robust IM} has been studied \cite{chen2016robust, he2016robust, ohsaka2017portfolio, anari2019structured}, where 
the goal is to maximize the influence spread over the entire set of possible model instances. 
A particular line of research is where one assumes that each node in the diffusion graph has some features encoding information about it (e.g., age or follower count in a social network). Then, the influence weight between a pair of nodes is a function of their features and a global low-dimensional \emph{hyperparameter} \cite{vaswani2017model, kalimeris2018learning}. This \emph{hyperparametric model} is  intuitive (user characteristics can play an important role in users influencing each other \cite{cialdini2001harnessing, platow2005s}) and computationally light when inferring probabilities, due to a smaller search space (given node features, once we find the right hyperparameter value, all probabilities are known). The inherent uncertainty in model instances arises from the potential values of the hyperparameter, which may not be known exactly, but only within a certain range.

Secondly, most existing IM algorithms are  designed for a \emph{static network}. Yet, in practice, networks are rarely  static. 
 For example, in microblogging (e.g., X a.k.a. Twitter), users may join and leave, may also create connections or remove them, at any time. 
 Similarly, in ad-hoc communication networks between mobile devices (e.g., drones),  which must perform synchronization tasks regularly,  communication links get  made or broken as the nodes move about in an urban landscape. 
 In such scenarios, an IM solution may be required frequently and must be readily available at any given moment. Recent works  \cite{ohsaka2016dynamic, wang2017real, peng2021dynamic}  studied \emph{dynamic IM}, i.e., they  aim to maintain a seed set with high influence spread over time, with just incremental computation per update to the network. Among these, \cite{peng2021dynamic} is the only one to (i) propose a solution whose amortized running time matches that of the SOTA offline IM algorithms (with a poly-logarithmic overhead in the number of nodes) and (ii) provide approximation guarantees on the quality of the returned seed set. 

We address both shortcomings of existing IM algorithms by studying the problem of \textit{robust influence maximization in dynamic diffusion networks}. To our knowledge, we are the first to study this problem, featuring the combination of robustness and dynamicity. A practical challenge in learning any such diffusion model from available cascades is that available data may be limited. Thus, the learning algorithm should have a low sample complexity. 
We adopt the hyperparametric IC model of \cite{ kalimeris2019robust}, which enables us to \eat{further reduce the gap between formal diffusion models and most real-world scenarios, where any realistic method for learning probabilities from cascades should have a low sample-complexity.} meet this requirement. Specifically, we aim to answer the following question: \emph{Under the hyperparametric model, is there a solution to robust IM over dynamic diffusion  networks which, at every time step, can return a seed set with quality guarantees while being significantly more efficient than running IM from scratch?} To this end, we design an approach, called \textbf{R}obust \textbf{I}nfluence \textbf{M}aximization over  \textbf{E}volving networks (\textbf{\algoname}), which runs multiplicative weight updates (MWU) \cite{chen2017robust} on values sampled from the hyperparameter space. For each weight combination, it computes the required seed set by running a greedy algorithm. We give theoretical quality guarantees for this solution and 
empirically show that \algoname\ returns a  good quality seed set much more efficiently than running IM from scratch upon every network update.

\spara{Potential Application Scenarios.} \eat{We view robustness to uncertain diffusion probabilities and dynamicity as two essential aspects of practical IM applications. While our approach remains application-agnostic and conceptually fits any IM setting,} We  highlight three concrete  applications where both uncertainty and dynamicity are prevalent. (1) In preventative health intervention  \cite{DBLP:conf/atal/WilderOHLWPWTR18}, algorithmic techniques are crucial to optimize targeting and delivery within a dynamic diffusion medium that blends virtual social links and real-world relationships. 
\eat{, both inherently subject to uncertainty and dynamicity.} (2) For combating misinformation on social media \cite{budak2011limiting,DBLP:conf/sigmod/Lakshmanan22} through targeted interventions, we need to quickly identify seed sets that are effective. The scale of social networks presents an added challenge, which can be addressed by leveraging sparse network representation \cite{DBLP:conf/kdd/MathioudakisBCGU11} which offers significant speedup in reaching large user groups while effectively mitigating the spread of misinformation. (3) In  synchronizing an ad-hoc communication network over a fleet of drones, links are inherently dynamic.
Drone (UAV) networks consist of drones as nodes, with one-to-one communications links between them  representing the diffusion edges. In such an ad-hoc communication network, a drone is more likely to communicate successfully with those in its vicinity, in a more or less reliable manner depending on distance, position in an urban landscape (e.g., no-fly or no-communication zones), intrinsic capabilities and resources, etc.  Such networks are indeed dynamic: at any time step, nodes may be added or removed (cut out from the world), and communication links between nodes can essentially appear and disappear as well. Also, the exact diffusion weights (e.g., communication reliability) between nodes may be difficult to assess accurately. In this setting, one application scenario we consider is the one of maintaining an overlay of ``cloud drones'' (e.g., high altitude drones) that can locate themselves and can also communicate with the ``basic drones'' (city drones) scattered in the city, in order to assist them in maintaining key information (synchronization, location, state-of-the-world, warnings, etc). Cloud drones may be paired with one or a few underlying basic drones, and by changing location they can change their corresponding (pairing) city drones.  From a dynamic IM perspective, the seed nodes to be maintained  at each time step are the basic drones chosen to be paired with cloud drones above. A diffusion would be initiated at these seed nodes,  and eventually reach as many other nodes as possible, through the ad-hoc communication network, by a viral diffusion process. Therefore, our method for dynamic and robust influence maximization can be applied for the objective of maximizing the ``reach'' of the cloud drones, possibly with a fail-back mechanism for the drones not reached by the viral diffusion.


\spara{Organization \& Contributions.} We review  classic IM  (\S\ref{sec:influence}), and  discuss its robust (\S\ref{sec:robust}) and dynamic (\S\ref{sec:dynamic}) variants. \eat{In the process, we introduce some important notations and concepts used in the remainder of the paper.} \eat{After this,} We formally state our problem, along with some basic properties and results (\S\ref{sec:problem}). In \S\ref{sec:solution}, we propose an approximate solution. For the \textit{incremental} case, where  nodes and edges can only be added to the network, 
our algorithm (\S\ref{sec:algorithm}) is a fusion of the MWU method HIRO from \cite{kalimeris2019robust} with  \cite{peng2021dynamic}'s greedy algorithm. Combining these algorithms is highly non-trivial and is one of our key contributions. \algoname\ enjoys  provable  approximation guarantees (\S\ref{sec:analysis}) going beyond mere extensions of those in \cite{kalimeris2019robust} and \cite{peng2021dynamic}.
\eat{
We stress that these guarantees are not trivial extensions of those in \cite{kalimeris2019robust} and \cite{peng2021dynamic}. }
%
%
We extend \algoname\ for the \textit{general} case, where nodes / edges can also be  removed, 
 and establish the first known theoretical approximation guarantees for this case.  

We empirically evaluate \algoname, comparing with the baseline  which runs from scratch after every network update. On both synthetic and real-world networks, it outperforms the baseline  w.r.t. running time by several orders of magnitude, while returning a solution of comparable quality.

To sum up, we are the first to tackle the combination of  robustness and dynamicity in IM, introducing \algoname, the first algorithm for \emph{robust dynamic IM}. \algoname\ is adaptable and flexible in handling these two crucial aspects. It can handle any degree of model uncertainty, from complete knowledge to scenarios where the hyperparameter lies within a $d$-cube of a given size. In the absence of uncertainty, it seamlessly aligns with the incremental update algorithm from \cite{peng2021dynamic}, while surpassing its capabilities in fully dynamic settings. Furthermore, when dynamicity is not a concern, it gracefully reduces to the HIRO method for robust IM. 

%% file: related.tex
\section{Related Work}
\spara{Influence Maximization (IM).} \cite{DR01} was the first work on IM, with \cite{KKT03} later elegantly formulating the problem along with diffusion models like IC, while also proving the submodularity of the influence spread function. A plethora of works have since focused on efficient IM approximations. A breakthrough in this aspect was achieved by the concept of reverse influence sampling, introduced by \cite{borgs2014maximizing} and later made practical in \cite{TSX15, nguyen2016stop}. Several works also focus on different variants of IM, like competitive IM \cite{BKS07, CNWZ07, HSCJ12, LL15, KLRS21} and opinion maximization \cite{GionisTT13, abebe2018opinion, saha2023voting}. We focus on the robust and dynamic IM variants, prior works on which are discussed next.

\spara{Robust IM.} 
Many recent works aim to devise a solution to the IM problem which is robust to noise and uncertainty in the diffusion model.  \cite{chen2017robust, anari2019structured} aim to find a seed set that maximizes the minimum spread across the set of possible models. However, they assume that the number of possible models is polynomial in the size of the input, which is often not true in practice. \cite{he2016robust, chen2016robust} focus on finding the seed set maximizing the robust ratio, which is the minimum ratio (across all possible models) of the influence spread of a seed set to the optimal influence spread. \cite{kalimeris2019robust} shows that the robust ratio is not a good metric to maximize, as the seed set maximizing it can have a minimum influence spread (across all possible models) a lot worse than the corresponding optimal value. Thus, it focuses on maximizing the minimum influence spread, under the hyperparametric diffusion model, assuming that the network is static.

\spara{Dynamic IM.} \cite{chen2015influential, ohsaka2016dynamic, wang2017real, aggarwal2012influential, liu2017shoulders, peng2021dynamic} are among the works that focus on IM in dynamic networks.  \cite{chen2015influential} proposes an upper bound interchange method with an approximation ratio of $0.5$. \cite{ohsaka2016dynamic} achieves a $(1 - 1/e)$-approximation via an algorithm based on reverse influence sampling. \cite{wang2017real}, inspired by streaming submodular maximization, designs an algorithm that maintains a constant approximate solution and is efficient in practice. However, \eat{all these works provide heuristic solutions and} none of them has rigorous theoretical guarantees on the amortized running time. \cite{peng2021dynamic} is the first work to provide such guarantees, albeit under the assumption that the influence weights are known accurately. However, it does not provide any empirical evaluation of its proposed solution.

%% file: preliminaries.tex
\section{Preliminaries}
\label{sec:prelims}
A diffusion network is a graph $G = (V, E)$ where $V$ is the set of nodes  and $E$ represents the connections between them. We denote $|V|$ and $|E|$ by $n$ and $m$ respectively.

\subsection{Diffusion Models and Influence Functions}
\label{sec:influence}

\spara{Independent Cascade (IC) Model.} To describe diffusion between nodes in a network, various models have been used in the literature. A commonly used one is  IC, which describes a discrete-step stochastic process through which information spreads in a network, from a set of initially active nodes. Each node can be active or inactive and each edge $e \in E$ in the network is associated with some probability $p_e$. All nodes begin as inactive. At time step $t = 0$, a subset of nodes $S \subseteq V$ (the \emph{seeds}) is chosen and becomes active. At each step $t + 1$, every node $u$ that became active at  step $t$ samples to influence each of its non-active neighbors $v$, independently with probability $p_{(u,v)}$. Activations are irreversible; a diffusion ends when there are no new activations. 

\spara{Influence Functions.} For a graph $G = (V, E)$ and vector of edge probabilities $\mathbf{p} \in [0, 1]^m$, the influence function $\sigma_{\mathbf{p}} : 2^V \rightarrow \mathbb{R}$ measures the expected number of nodes that will become influenced in the graph $G$ for a seed set $S \subseteq V$:
\begin{small}
\begin{equation*}
    \sigma_{\mathbf{p}}(S) = \sum_{A \subseteq E} r_A(S) \prod_{e \in A} p_e \prod_{e \notin A} \left( 1 - p_e \right)
\end{equation*}
\end{small}
\noindent where $r_A(S)$ denotes the number of nodes  reachable in $G$ from $S$ using only edges from $A$. An important property of $\sigma_{\mathbf{p}}$ is that it is non-decreasing and submodular\footnote{\eat{In short, }A set function $f$ is non-decreasing if \eat{adding an element to a set can never decrease the output of $f$ on it }for any set $A$ and $x \notin A$, $f(A \cup \{x\}) \geq f(A)$. Furthermore, $f$ is said to be submodular if \eat{adding an element to a set provides less additional benefit when the set already contains some other elements; in other words, }for any $A \subseteq B$ and $x \notin B$, $f(A \cup \{x\}) - f(A) \geq f(B \cup \{x\}) - f(B)$.} for any $\mathbf{p}$. 

\spara{Influence Maximization (IM).} Given a network $G$, a probability vector $\mathbf{p}$, and a limited budget of $k$ seed nodes, the classic IM problem aims to find the seed set
\begin{small}
\begin{equation*}
    S^* = \argmax_{S \subseteq V : |S| \leq k} \sigma_{\mathbf{p}} \left( S \right)
\end{equation*}
\end{small}
In \S\ref{sec:robust} we discuss robust IM, where the probability vector is not fixed and instead is assumed to belong to a given set of vectors. In \S\ref{sec:dynamic} we discuss dynamic IM, where the diffusion network can change with time.

\subsection{Robust Influence Maximization}
\label{sec:robust}
Given a diffusion network and a \textit{set} of possible edge probability vectors, the robust IM problem aims to find a seed set of size $k$ that has high spread value for every possible influence function that corresponds to a possible probability vector. One way of defining the set of possible probability vectors is via a hyperparametric model \cite{kalimeris2019robust}. In this model, each node is associated with a vector of features, encoding information about it. The feature vector of each edge $e = (u, v)$, denoted $x_e \in X \subseteq [-1, 1]^d$, is obtained by the concatenation of the $d/2$-dimensional feature vectors of its endpoints $u$ and $v$.  The diffusion probability of the edge $e$ is then a function of its feature vector $x_e$ and a global hyperparameter $\theta \in \Theta = [-B, B]^d$ for some constant radius $B \in \mathbb{R}^+$. In other words, there exists a function $H : \Theta \times X \rightarrow [0, 1]$ such that $p_e = H(\theta, x_e)$. In that case, the set of possible probability vectors is given by $\mathcal{H} = \left\{ \left(H(\theta, x_e)\right)_{e \in E} : \theta \in \Theta \right\}$. Given a set $\mathcal{H}$ of probability vectors,  robust IM under the hyperparametric model aims to find the seed set
\begin{small}
\begin{equation*}
    S^* = \argmax_{S \subseteq V : |S| \leq k} \min_{\mathbf{p} \in \mathcal{H}} \sigma_{\mathbf{p}} \left( S \right)
\end{equation*}
\end{small}

In this paper, we focus on generalized linear hyperparametric models that are $1$-Lipschitz with respect to the $L_1$ norm, following those used in the IM literature \cite{vaswani2017model, kalimeris2018learning}. For such models, we can reduce robust continuous optimization over $\Theta$ to a robust discrete one over some sampled hyperparameter values. Examples of such models include linear with $H(\theta, x_e) = \theta^T x_e$, logistic with $H(\theta, x_e) = 1/ [1 + \exp(- \theta^T x_e)]$, and probit with $H(\theta, x_e) = \Phi(\theta^T x_e)$, where $\Phi$ is the CDF of the standard Gaussian distribution.

\subsection{Dynamic Diffusion Networks}
\label{sec:dynamic}

In diffusion scenarios, networks are mostly non-static, they can change with time in a variety of ways. As in \cite{peng2021dynamic}, we assume that at any time step $t$ one of the following changes may occur: (1) a node is inserted, (2) an edge is inserted, (3) an edge is removed, or (4) a node is removed. In the \emph{incremental setting}, only the first two types of changes are allowed. In other words, at each time step, a node or an edge can only be inserted to (not removed from) the network. In a \emph{fully dynamic setting}, all  types of changes are allowed.

Let $G_t = (V_t, E_t)$ denote the diffusion network at time $t$. Given a probability vector $\mathbf{p}$, the dynamic IM problem aims to find, for each time step $t$, the seed set
\begin{small}
\begin{equation*}
    S_t^* = \argmax_{S_t \subseteq V_t : |S_t| \leq k} \sigma_{\mathbf{p}} \left( S_t \right)
\end{equation*}
\end{small}
\textbf{Remark:} As in previous works, we assume that time is discretized, i.e., 
the actual times at which network updates occur can be mapped to discrete time steps.

\subsection{Problem Statement and Hardness}
\label{sec:problem}

We are now ready to formally state our problem of robust influence maximization in dynamic social networks.

\begin{problem}
\label{prob:rob_dyn_im}
    For a  budget $k$ and a dynamic diffusion network whose instance at time $t$ is $G_t = (V_t, E_t)$  and $\mathcal{H}_t = \left\{ \left( H \left( \theta, x_e \right) \right)_{e \in E_t} : \theta \in \Theta \right\}$ is the set of possible edge probability configurations, find for all time steps $t$ the seed set
    \begin{small}
    \begin{equation*}
        S_t^* = \argmax_{S_t \subseteq V_t : |S_t| \leq k} \min_{\mathbf{p} \in \mathcal{H}_t} \sigma_{\mathbf{p}, t} \left( S_t \right)
    \end{equation*}
    \end{small}
    where $\sigma_{\mathbf{p}, t} \left( S_t \right)$ is the expected influence spread of $S_t$ in $G_t$ under the edge probability configuration $\mathbf{p}$.
\end{problem}
\noindent Clearly, hardness results similar to those of robust IM or  dynamic IM also hold here. Specifically, 
at each time step, it is \NP-hard to (1) get any approximation better than $(1 - 1/e)$ and (2) find a seed set of size $k$ that obtains any approximation better than $\bigO(\log n)$ \cite{kalimeris2019robust}. So we look for an $(\alpha, \beta)$ bi-criteria approximation, where $\alpha, \beta \in (0, 1)$; i.e. a seed set $\hat{S}_t$ s.t. $\beta | \hat{S}_t | \leq k$ and
\begin{small}
\begin{equation*}
    \min_{\mathbf{p} \in \mathcal{H}_t} \sigma_{\mathbf{p}, t} \left( \hat{S}_t \right) \geq \alpha \max_{S_t \subseteq V_t : |S_t| \leq k} \min_{\mathbf{p} \in \mathcal{H}_t} \sigma_{\mathbf{p}, t} \left( S_t \right)
\end{equation*}
\end{small}

When the hyperparametric model and the set of features are clear from the context, the probability vector $\mathbf{p}$ is defined entirely by the hyperparameter $\theta$. In such cases, for ease of understanding, we shall use $\sigma_{\theta, t}$ instead of $\sigma_{\mathbf{p}, t}$.

%% file: solution.tex
\section{Solution}
\label{sec:solution}
In this section, we propose an approximate solution to Problem \ref{prob:rob_dyn_im} for both the incremental and the fully dynamic settings. We discuss the overall algorithm in \S\ref{sec:algorithm} and conduct a theoretical analysis of its quality and running time in \S\ref{sec:analysis}. 

\begin{algorithm}[t!]
	\caption{\textsc{Restart}}
	\begin{algorithmic}[1]
        \STATE $\delta \gets n_0 \left( \frac{\delta_1}{(16(2n_0)^{Tk})} \right) ^{\frac{12n_0^2}{k^3}}$
        \STATE $R \gets k\epsilon_1^{-2} \log \left( \frac{n_0}{\delta} \right)$
        \FOR{$i = 1$ \TO $l$} \label{line:begin_p}
        \STATE $K_i \gets 0$
        \REPEAT
        \STATE Sample RR set with hyperparameter $\theta_i$ uniformly over $V_t$
        \STATE $K_i \gets K_i + 1$
        \UNTIL{$Rm_0$ edges are checked}
        \STATE $p_i \gets \frac{K_i}{n_0}$
        \ENDFOR \label{line:end_p}
        \FORALL{$(u, v) \in V \times V$}
        \STATE $I(u, v) \gets \emptyset$
        \ENDFOR
        \FOR{$i = 1$ \TO $l$} \label{line:begin_gen_rr}
        \FORALL{$v \in V$}
        \STATE Generate RR set $R_{v,est,i}$ with probability $p_i$
        \STATE Generate RR set $R_{v,cv,i}$ with probability $p_i$
        \FORALL{$u \in R_{v,cv,i}$}
        \STATE $I(u, v) \gets I(u, v) \cup \{i\}$
        \ENDFOR
        \ENDFOR
        \ENDFOR \label{line:end_gen_rr}
        \FORALL{$(u, v) \in V \times V$} \label{line:begin_find}
        \FORALL{$i \in I(u, v)$}
        \STATE $E_{cv, i} \gets E_{cv, i} \cup \{(u, v)\}$
        \ENDFOR
        \STATE \textsc{Find-Seeds}($(u, v)$)
        \ENDFOR \label{line:end_find}
	\end{algorithmic}
	\label{alg:restart}
\end{algorithm}

\begin{algorithm}[t!]
	\caption{\textsc{Find-Seeds}}
	\begin{algorithmic}[1]
		\REQUIRE $(u', v')$
		\ENSURE Seed set
        \FOR{$j = 1$ \TO $T$}
        \FOR{$i = 1$ \TO $l$}
        \STATE $w_j[i] \propto \exp \left( -\eta \sum_{\tau=1}^{j-1} f_{cv,i} (S^{(\tau)}) \right)$
        \ENDFOR
        \STATE $S^{(j)} \gets$ Greedy($\langle w_j, f_{cv} \rangle, (u', v')$) \label{line:greedy}
        \ENDFOR
        \FOR{$i = 0$ \TO $\lceil \epsilon^{-1} \log n \rceil$}
        \STATE $S_i \gets$ Seed set computed by \textsc{Greedy} for thread $i$ with uniform weights
        \ENDFOR
        \RETURN Uniform distribution over $\left\{ S^{(1)}, \ldots, S^{(T)} \right\}$
	\end{algorithmic}
	\label{alg:find_seeds}
\end{algorithm}

\subsection{Algorithm}
\label{sec:algorithm}
We first provide a brief overview of our method. Let $n_0$ and $m_0$ respectively denote the initial number of nodes and edges in our network. As mentioned in the last paragraph of \S~\ref{sec:robust}, we first sample $l$ values $\theta_1, \ldots, \theta_l$ from the hyperparameter space $\Theta$. Then, for each sampled hyperparameter, following the SOTA in IM, we sample a few reverse reachable (RR) sets and solve a maximum coverage problem on those sets, thereby computing the seed set for the initial network. After that, when a node or an edge is added to or removed from the network, we generally run some incremental computations, e.g. modifying the previously generated RR sets and hence the seed sets. However, if some conditions are satisfied after the network update (e.g., the size of the network is doubled or halved), we restart from scratch. In the next  paragraphs, we explain our method in more detail.

\eat{
\begin{figure*}[t!]
    \centering
     \hspace{0.5cm}
    \includegraphics[scale=0.255]{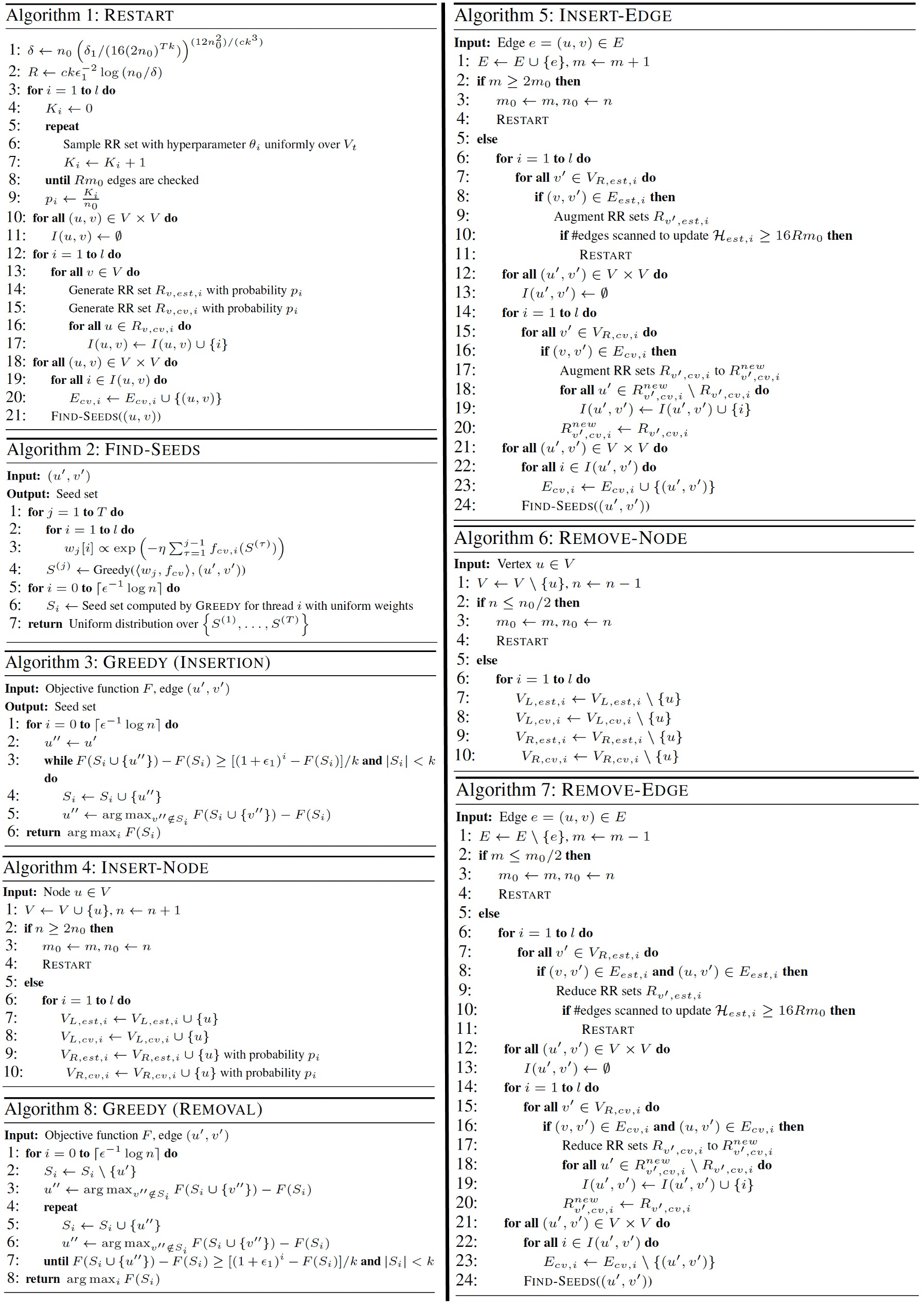}
\end{figure*}
}



\begin{algorithm}[t!]
	\caption{\textsc{Greedy (Insertion)}}
	\begin{algorithmic}[1]
		\REQUIRE Objective function $F$, edge $(u', v')$
		\ENSURE Seed set
        \FOR{$i = 0$ \TO $\lceil \epsilon^{-1} \log n \rceil$}
        \STATE $u'' \gets u'$
        \WHILE{$F(S_i \cup \{u''\}) - F(S_i) \geq [(1 + \epsilon_1)^i - F(S_i)] / k$ \AND $|S_i| < k$} \label{line:cond_k}
        \STATE $S_i \gets S_i \cup \{u''\}$
        \STATE $u'' \gets \argmax_{v'' \notin S_i} F(S_i \cup \{v''\}) - F(S_i)$
        \ENDWHILE
        \ENDFOR
        \RETURN $\arg\max_{i} F(S_i)$
	\end{algorithmic}
	\label{alg:greedy_insert}
\end{algorithm}

\begin{algorithm}[t!]
	\caption{\textsc{Insert-Node}}
	\begin{algorithmic}[1]
		\REQUIRE Node $u \in V$
        \STATE $V \gets V \cup \{u\}$, $n \gets n + 1$
        \IF{$n \geq 2n_0$}
        \STATE $m_0 \gets m$, $n_0 \gets n$
        \STATE \textsc{Restart}
        \ELSE
        \FOR{$i = 1$ \TO $l$}
        \STATE $V_{L,est,i} \gets V_{L,est,i} \cup \{u\}$
        \STATE $V_{L,cv,i} \gets V_{L,cv,i} \cup \{u\}$
        \STATE $V_{R,est,i} \gets V_{R,est,i} \cup \{u\}$ with probability $p_i$
        \STATE $V_{R,cv,i} \gets V_{R,cv,i} \cup \{u\}$ with probability $p_i$
        \ENDFOR
        \ENDIF
	\end{algorithmic}
	\label{alg:insert_node}
\end{algorithm}

When our method runs for the very first time, for each $i \in [l]$, it invokes Alg. \ref{alg:restart} -- \textsc{Restart} where, from each node $v$, it generates two RR sets $R_{v,est,i}$ and $R_{v,cv,i}$ (lines \ref{line:begin_gen_rr}-\ref{line:end_gen_rr}), each with probability $p_i$ (computed in lines \ref{line:begin_p}-\ref{line:end_p}). \footnote{Note that the value of $p_i$ computed in Alg. \ref{alg:restart} can be greater than $1$. In that case, from each node $v$, we generate $\lfloor p_i \rfloor$ RR sets with probability 1 and an additional RR set with probability $p_i - \lfloor p_i \rfloor$.}
These collections of RR sets are modelled as bipartite coverage graphs $\mathcal{H}_{est,i} = (V_{L,est,i}, V_{R,est,i}, E_{est,i})$ and $\mathcal{H}_{cv,i} = (V_{L,cv,i}, V_{R,cv,i}, E_{cv,i})$, where $(u, v) \in E_{est,i}$ (resp. $E_{cv,i}$) if and only if $u \in R_{v,est,i}$ (resp. $R_{v,cv,i}$).\footnote{We generate two bipartite graphs to decouple the correlation between estimating the influence spread of a seed set and estimating the time required to sample an RR set \cite{peng2021dynamic}. $\mathcal{H}_{cv,i}$ is used for the former, while $\mathcal{H}_{est,i}$ is used for the latter.} Clearly, the number of RR sets $R_{v,cv,i}$ covered by a seed set $S$ is equal to the number of nodes in $V_{R,cv,i}$ adjacent to $S$ in $\mathcal{H}_{cv,i}$. This number, normalized (divided) by $p_i$, is denoted by $f_{cv,i}(S)$ and is used to estimate the true influence spread $\sigma_{\theta_i}(S)$ of $S$ for hyperparameter $\theta_i$.

\begin{algorithm}[t!]
	\caption{\textsc{Insert-Edge}}
	\begin{algorithmic}[1]
		\REQUIRE Edge $e = (u, v) \in E$
        \STATE $E \gets E \cup \{e\}$, $m \gets m + 1$
        \IF{$m \geq 2m_0$}
        \STATE $m_0 \gets m$, $n_0 \gets n$
        \STATE \textsc{Restart}
        \ELSE
        \FOR{$i = 1$ \TO $l$}
        \FORALL{$v' \in V_{R,est,i}$}
        \IF{$(v, v') \in E_{est,i}$}
        \STATE Augment RR sets $R_{v',est,i}$
        \IF{\#edges scanned to update $\mathcal{H}_{est,i} \geq 16Rm_0$}
        \STATE \textsc{Restart}
        \ENDIF
        \ENDIF
        \ENDFOR
        \ENDFOR
        \FORALL{$(u', v') \in V \times V$}
        \STATE $I(u', v') \gets \emptyset$
        \ENDFOR
        \FOR{$i = 1$ \TO $l$}
        \FORALL{$v' \in V_{R,cv,i}$}
        \IF{$(v, v') \in E_{cv,i}$}
        \STATE Augment RR sets $R_{v',cv,i}$ to $R_{v',cv,i}^{new}$
        \FORALL{$u' \in R_{v',cv,i}^{new} \setminus R_{v',cv,i}$}
        \STATE $I(u', v') \gets I(u', v') \cup \{i\}$
        \ENDFOR
        \STATE $R_{v',cv,i}^{new} \gets R_{v',cv,i}$
        \ENDIF
        \ENDFOR
        \ENDFOR
        \FORALL{$(u', v') \in V \times V$}
        \FORALL{$i \in I(u', v')$}
        \STATE $E_{cv, i} \gets E_{cv, i} \cup \{(u', v')\}$
        \ENDFOR
        \STATE \textsc{Find-Seeds}($(u', v')$)
        \ENDFOR
        \ENDIF
	\end{algorithmic}
	\label{alg:insert_edge}
\end{algorithm}

We now discuss how to find the seeds via dynamic maximum coverage on the $\mathcal{H}_{cv,i}$s. As shown in lines \ref{line:begin_find}-\ref{line:end_find} of Alg. \ref{alg:restart}, whenever an edge is added to the $\mathcal{H}_{cv,i}$s (i.e., a node is added to the corresponding RR sets), the seed set also needs to be updated by calling Alg. \ref{alg:find_seeds} -- \textsc{Find Seeds}. Along the lines of HIRO, 
Alg. \ref{alg:find_seeds} runs $T$ iterations of multiplicative weight updates (MWU), each computing a candidate seed set after running an approximation algorithm to maximize a convex combination $F$ of the sampled influence functions (line \ref{line:greedy});  higher weights are given to those with poor performance in previous iterations.
Similar to \cite{peng2021dynamic}, the approximation algorithm used in line \ref{line:greedy} of Alg. \ref{alg:find_seeds}, which is shown in Alg. \ref{alg:greedy_insert} -- \textsc{Greedy (Insertion)}, maintains a seed set for each of the $\lceil \epsilon_1^{-1} \log n \rceil$ threads, where thread $i$ corresponds to the estimate $(1 + \epsilon_1)^i$ of the maximum value of $F$. Alg. \ref{alg:greedy_insert} adds to each thread's seed set any node with marginal gain above a threshold. After that, it returns the seed set maximizing $F$. In the end, after getting the candidate seed sets using MWU, Alg. \ref{alg:find_seeds} returns one of those sets uniformly at random, after updating each thread's seed set to be the last one computed by Alg. \ref{alg:greedy_insert}, with $F$ having uniform weights.

After the initial run, when a node (resp. edge) is inserted to the network, Alg. \ref{alg:insert_node} -- \textsc{Insert Node} (resp. \ref{alg:insert_edge} -- \textsc{Insert Edge}) is invoked. As shown in the pseudocode, the entire process restarts from scratch (invokes Alg. \ref{alg:restart}) when the number of nodes or edges is doubled. Unless the restarting condition is met, the only computation required when a node is inserted is creating some RR sets containing only that node. When an edge is inserted, however, it can enhance the information diffusion in the network, thus more updates are needed. Specifically, whenever an edge $(u,v)$ with probability $p_{(u,v)}$ is added, each RR set containing $v$ but not $u$ should be augmented by adding $u$ with probability $p_{(u,v)}$. If it is added, we continue augmenting this RR set, i.e., adding new nodes to the set by sampling incoming edges of $u$ and \eat{so on} repeating the same for the new ones till no further node can be reached. For RR sets in $\mathcal{H}_{est,i}$, if this process of augmentation results in the overall number of edges in $E$ scanned while updating $\mathcal{H}_{est}$ exceeding a threshold, we restart from scratch (Alg. \ref{alg:restart}). For RR sets in $\mathcal{H}_{cv,i}$, augmentation clearly leads to new (bipartite) edges being added, which means we need to update the seed set (Alg. \ref{alg:find_seeds}) for each such addition.

\begin{algorithm}[t!]
	\caption{\textsc{Remove-Node}}
	\begin{algorithmic}[1]
		\REQUIRE Vertex $u \in V$
        \STATE $V \gets V \setminus \{u\}$, $n \gets n - 1$
        \IF{$n \leq n_0 / 2$}
        \STATE $m_0 \gets m$, $n_0 \gets n$
        \STATE \textsc{Restart}
        \ELSE
        \FOR{$i = 1$ \TO $l$}
        \STATE $V_{L,est,i} \gets V_{L,est,i} \setminus \{u\}$
        \STATE $V_{L,cv,i} \gets V_{L,cv,i} \setminus \{u\}$
        \STATE $V_{R,est,i} \gets V_{R,est,i} \setminus \{u\}$
        \STATE $V_{R,cv,i} \gets V_{R,cv,i} \setminus \{u\}$
        \ENDFOR
        \ENDIF
	\end{algorithmic}
	\label{alg:remove_node}
\end{algorithm}

\begin{algorithm}[t!]
	\caption{\textsc{Remove-Edge}}
	\begin{algorithmic}[1]
		\REQUIRE Edge $e = (u, v) \in E$
        \STATE $E \gets E \setminus \{e\}$, $m \gets m - 1$
        \IF{$m \leq m_0 / 2$}
        \STATE $m_0 \gets m$, $n_0 \gets n$
        \STATE \textsc{Restart}
        \ELSE
        \FOR{$i = 1$ \TO $l$}
        \FORALL{$v' \in V_{R,est,i}$}
        \IF{$(v, v') \in E_{est,i}$ \AND $(u, v') \in E_{est,i}$}
        \STATE Reduce RR sets $R_{v',est,i}$
        \IF{\#edges scanned to update $\mathcal{H}_{est,i} \geq 16Rm_0$}
        \STATE \textsc{Restart}
        \ENDIF
        \ENDIF
        \ENDFOR
        \ENDFOR
        \FORALL{$(u', v') \in V \times V$}
        \STATE $I(u', v') \gets \emptyset$
        \ENDFOR
        \FOR{$i = 1$ \TO $l$}
        \FORALL{$v' \in V_{R,cv,i}$}
        \IF{$(v, v') \in E_{cv,i}$ \AND $(u, v') \in E_{cv,i}$}
        \STATE Reduce RR sets $R_{v',cv,i}$ to $R_{v',cv,i}^{new}$
        \FORALL{$u' \in R_{v',cv,i}^{new} \setminus R_{v',cv,i}$}
        \STATE $I(u', v') \gets I(u', v') \cup \{i\}$
        \ENDFOR
        \STATE $R_{v',cv,i}^{new} \gets R_{v',cv,i}$
        \ENDIF
        \ENDFOR
        \ENDFOR
        \FORALL{$(u', v') \in V \times V$}
        \FORALL{$i \in I(u', v')$}
        \STATE $E_{cv, i} \gets E_{cv, i} \setminus \{(u', v')\}$
        \ENDFOR
        \STATE \textsc{Find-Seeds}($(u', v')$)
        \ENDFOR
        \ENDIF
	\end{algorithmic}
	\label{alg:remove_edge}
\end{algorithm}

The techniques for handling  node and edge removal are shown in Alg. \ref{alg:remove_node} -- \textsc{Remove-Node} and \ref{alg:remove_edge} -- \textsc{Remove-Edge}. In both cases, a process restarts from scratch (Alg. \ref{alg:restart}) whenever the number of nodes or edges is halved. A node removal only leads to its removal from all RR sets. When an edge $(u, v)$ is removed, however, the diffusion pathways can be affected, and thus more processing is needed. Specifically, we need to reduce each RR set containing both $u$ and $v$, where $u$ was reached after sampling incoming edges to $v$, by removing $u$ and all of its descendants in the RR set.\footnote{For this, we need to store the graph structure of each RR set while it is generated.} For RR sets in $\mathcal{H}_{est,i}$, if this process results in the overall number of edges in $E$ scanned while updating $\mathcal{H}_{est}$ exceeding\footnote{Removing a node and all of its descendants from an RR set requires scanning the edges in the graph structure of the RR set, which leads to the increase in the number of edges in the network scanned while updating the corresponding coverage graph.} a threshold, we restart from scratch (Alg. \ref{alg:restart}). For RR sets in $\mathcal{H}_{cv,i}$, such an RR set reduction leads to the removal of (bipartite) edges, which requires an update of the seed set. The seed set is computed by calling the same MWU-based  method as for insertions (Alg. \ref{alg:find_seeds}). However, the approximation algorithm called in line \ref{line:greedy} of Alg. \ref{alg:find_seeds} is different for edge removals. When an edge $(u',v') \in E_{cv,i}$ is to be removed, for each thread's seed set, we remove $u'$ and add the node with the largest marginal gain (possibly $u'$); after that, we add each node with marginal gain above a threshold. This is shown in Alg. \ref{alg:greedy_remove} -- \textsc{Greedy (Removal)}.

\begin{algorithm}[t!]
	\caption{\textsc{Greedy (Removal)}}
	\begin{algorithmic}[1]
		\REQUIRE Objective function $F$, edge $(u', v')$
        \FOR{$i = 0$ \TO $\lceil \epsilon^{-1} \log n \rceil$}
        \STATE $S_i \gets S_i \setminus \{u'\}$
        \STATE $u'' \gets \argmax_{v'' \notin S_i} F(S_i \cup \{v''\}) - F(S_i)$
        \REPEAT
        \STATE $S_i \gets S_i \cup \{u''\}$
        \STATE $u'' \gets \argmax_{v'' \notin S_i} F(S_i \cup \{v''\}) - F(S_i)$
        \UNTIL{$F(S_i \cup \{u''\}) - F(S_i) \geq [(1 + \epsilon_1)^i - F(S_i)] / k$ \AND $|S_i| < k$}
        \ENDFOR
        \RETURN $\arg\max_{i} F(S_i)$
	\end{algorithmic}
	\label{alg:greedy_remove}
\end{algorithm}

Note that the method in \cite{peng2021dynamic} needs only two coverage graphs $\mathcal{H}_{est}$ and $\mathcal{H}_{cv}$, while ours requires $l$ such pairs of graphs to account for the noise in the edge probabilities. Also, \cite{peng2021dynamic} finds seeds by invoking the greedy algorithm for maximizing coverage only once for each update to $\mathcal{H}_{cv}$, whereas to account for the edge probability noise, we need to do the same via $T$ MWU iterations, and in each iteration, we need to maximize a convex combination of $l$ coverage functions. Fortunately, we can still provide quality guarantees for the resultant method, as shown in \S\ref{sec:analysis}.

\subsection{Analysis}
\label{sec:analysis}

We now analyze \algoname's quality and running time. Our key result is stated in Theorem~\ref{th:i2theta}, to which we lead up by establishing some important intermediate results.

We first give a quality guarantee for our greedy algorithm. 
Consider a run of the algorithm upon addition of a node to an RR set (so a call to Alg. \ref{alg:find_seeds}), with the input objective function $F$ having optimal value $OPT$.
For any seed set $S$ and $x < |S|$, let $S_{1:x}$ be the set of the first $x$ nodes added to $S$. Also, for $x < y < |S|$, let $F_{t_y}(S_{1:x})$ be the value of $F(S_{1:x})$ at the time the $y^{\text{th}}$ node was added to $S$. For completeness, define $F_{t_{k+1}}(\cdot) = F(\cdot)$.

\begin{lemma}
\label{lem:greedy}
    Our greedy algorithm returns a seed set $\hat{S}$ such that $F\left(\hat{S}\right) \geq (1 - 1/e - \gamma - \epsilon_1) \cdot OPT$, 
    where
    \begin{small}
    \begin{gather*}
       \gamma = 
       \begin{cases}
           \max_{x=1}^{\left| \hat{S} \right|} \frac{F_{t_x} \left( \hat{S}_{1:x} \right) - F_{t_{x+1}} \left( \hat{S}_{1:x} \right)}{OPT} \left( 1 - \frac{1}{k} \right)^{\left| \hat{S} \right| - x} \quad \text{incremental} \\
           \sum_{x=1}^{\left| \hat{S} \right|} \frac{F_{t_x} \left( \hat{S}_{1:x} \right) - F_{t_{x+1}} \left( \hat{S}_{1:x} \right)}{OPT} \left( 1 - \frac{1}{k} \right)^{\left| \hat{S} \right| - x} \quad \text{fully dynamic}
       \end{cases}
   \end{gather*}
    \end{small}
\end{lemma}

\begin{proof}
Since we have $\lceil \epsilon^{-1} \log n \rceil$ threads, there exists a thread $j$ such that
\begin{equation}
    (1 + \epsilon_1)^j \leq OPT < (1 + \epsilon_1)^{j + 1} \label{eq:j}
\end{equation}
It suffices to prove that thread $j$ computes a seed set $S_j$ satisfying the quality guarantee, since the seed set returned is $\argmax_{i \in \left[ \lceil \epsilon^{-1} \log n \rceil \right]} F(S_i)$. There are 2 cases: $\left| S_j \right| = k$ and $\left| S_j \right| < k$. 

When $\left| S_j \right| < k$; we show that $F_t(S_j) \geq (1 + \epsilon_1)^j \geq (1 + \epsilon_1)^{-1} OPT$. We prove by contradiction; assuming $F(S_j) < (1 + \epsilon_1)^j$, we derive $OPT < (1 + \epsilon_1)^j$, which is not true from \eqref{eq:j}. Let $O$ denote the optimal seed set. Then, for all $o \in O$, $F(S_j \cup \{o\}) - F(S_j) < \frac{1}{k} \left[ (1 + \epsilon_1)^j - F(S_j) \right]$; this is guaranteed by our algorithm for $o \notin S_j$, and for $o \in S_j$, $F(S_j \cup \{o\}) - F(S_j) = 0 < \frac{1}{k} \left[ (1 + \epsilon_1)^j - F(S_j) \right]$. Thus,
\begin{small}
\begin{align*}
    OPT = F(O) &\leq F(S_j \cup O) = F(S_j) + F(S_j \cup O) - F(S_j) \\
    &\leq F(S_j) + \sum_{o \in O} \left[ F(S_j \cup \{o\}) - F(S_j) \right] \\
    &< F(S_j) + k \cdot \frac{1}{k} \left[ (1 + \epsilon_1)^j - F(S_j) \right] = (1 + \epsilon_1)^j
\end{align*}
\end{small}
and hence the contradiction.

Now, we consider the case $\left| S_j \right| = k$. For each $y \in [1, k]$, define $s_{j, y}$ as the $y$-th element added to $S_j$, and $S_{j, 1:y} = \left\{ s_{j, 1}, \ldots, s_{j, y} \right\}$. For completeness, define $S_{j, 0} = \emptyset$. 

\spara{Incremental setting.} Let $x$ denote the number of elements added to $S_j$ before the current call to Algorithm 2, i.e. $S_j = \left\{ s_{j, 1}, \ldots, s_{j, x}, s_{j, x + 1}, \ldots, s_{j, k} \right\}$. Algorithm 3 guarantees that
\begin{small}
\begin{gather*}
    F \left( S_{j, 1:y} \right) - F \left( S_{j, 1:y-1} \right) \geq \frac{1}{k} \left[ (1 + \epsilon_1)^j - F \left( S_{j, 1:y-1} \right) \right], \\ 
           y \in (x, k] \\
    F_{t_y} \left( S_{j, 1:y} \right) - F_{t_y} \left( S_{j, 1:y-1} \right) \geq \frac{1}{k} \left[ (1 + \epsilon_1)^j - F_{t_y} \left( S_{j, 1:y-1} \right) \right], \\ 
          y \in [0, x]
\end{gather*}
\end{small}
Thus we have
\begin{small}
\begin{align*}
    &(1 + \epsilon_1)^j - F \left( S_{j, 1:k} \right) \\
    &= (1 + \epsilon_1)^j - F \left( S_{j, 1:k-1} \right) + F \left( S_{j, 1:k-1} \right) - F \left( S_{j, 1:k} \right) \\
    & \leq (1 + \epsilon_1)^j - F \left( S_{j, 1:k-1} \right) - \frac{1}{k} \left[ (1 + \epsilon_1)^j - F \left( S_{j, 1:k-1} \right) \right] \\
    & = \left( 1 - \frac{1}{k} \right) \left[ (1 + \epsilon_1)^j - F \left( S_{j, 1:k-1} \right) \right] \\
    & \vdots \\
    & \leq \left( 1 - \frac{1}{k} \right)^{k-x} \left[ (1 + \epsilon_1)^j - F_{t_{x+1}} \left( S_{j, 1:x} \right) \right] \\
    & = \left( 1 - \frac{1}{k} \right)^{k-x} \left[ \left( (1 + \epsilon_1)^j - F_{t_x} \left( S_{j, 1:x} \right) \right) \right. \\
    & \qquad\qquad\qquad\qquad \left. + \left( F_{t_x} \left( S_{j, 1:x} \right) - F_{t_{x+1}} \left( S_{j, 1:x} \right) \right) \right] \\
    & \vdots \\
    & \leq \left( 1 - \frac{1}{k} \right)^k (1 + \epsilon_1)^j + \left( 1 - \frac{1}{k} \right)^{k-x} \left[ F_{t_x} \left( S_{j, x} \right) - F_{t_{x+1}} \left( S_{j, x} \right) \right]
\end{align*}
\end{small}
Thus we have
\begin{small}
\begin{align*}
    & F \left( S_j \right) \\
    & \geq \left[ 1 - \left( 1 - \frac{1}{k} \right)^k \right] (1 + \epsilon_1)^j - \left( 1 - \frac{1}{k} \right)^{k - x} \left[ F_{t_x} \left( S_{j, x} \right) - F_{t_{x+1}} \left( S_{j, x} \right) \right] \\
    & \geq \left( 1 - \frac{1}{e} \right) (1 + \epsilon_1)^{-1} OPT - \left( 1 - \frac{1}{k} \right)^{k - x} \left[ F_{t_x} \left( S_{j, x} \right) - F_{t_{x+1}} \left( S_{j, x} \right) \right] \\
    & \geq \left( 1 - \frac{1}{e} - \epsilon_1 \right) OPT - \max_{x \in [k]} \left( 1 - \frac{1}{k} \right)^{k - x} \left[ F_{t_x} \left( S_{j, x} \right) - F_{t_{x+1}} \left( S_{j, x} \right) \right]
\end{align*}
\end{small}

\spara{Fully dynamic setting.} We have
\begin{small}
\begin{align*}
    & (1 + \epsilon_1)^j - F \left( S_{j, k} \right) = (1 + \epsilon_1)^j - F_{t_k} \left( S_{j, k} \right) + \left[ F_{t_k} \left( S_{j, k} \right) - F \left( S_{j, k} \right) \right] \\
    & = (1 + \epsilon_1)^j - F_{t_k} \left( S_{j, k-1} \right) + F_{t_k} \left( S_{j, k-1} \right) - F_{t_k} \left( S_{j, k} \right) \\
    & \qquad + \left[ F_{t_k} \left( S_{j, k} \right) - F_{t_{k+1}} \left( S_{j, k} \right) \right] \\
    & \leq \left( 1 - \frac{1}{k} \right) \left[ (1 + \epsilon_1)^j - F_{t_k} \left( S_{j, k-1} \right) \right] + \left[ F_{t_k} \left( S_{j, k} \right) - F_{t_{k+1}} \left( S_{j, k} \right) \right] \\
    & = \left( 1 - \frac{1}{k} \right) \left[ (1 + \epsilon_1)^j - F_{t_k-1} \left( S_{j, k-1} \right) \right] + \left[ F_{t_k} \left( S_{j, k} \right) - F_{t_{k+1}} \left( S_{j, k} \right) \right] \\
    & \qquad \qquad + \left( 1 - \frac{1}{k} \right) \left[ F_{t_k-1} \left( S_{j, k-1} \right) - F_{t_k} \left( S_{j, k-1} \right) \right] \\
    & \vdots \\
    & \leq \left( 1 - \frac{1}{k} \right)^k (1 + \epsilon_1)^j + \sum_{x=1}^{k} \left( 1 - \frac{1}{k} \right)^{k - x} \left[ F_{t_x} \left( S_{j, x} \right) - F_{t_{x+1}} \left( S_{j, x} \right) \right]
\end{align*}
\end{small}
and by similar steps as in the incremental setting, we have
\begin{small}
\begin{align*}
    & F \left( S_j \right) \\
    & \geq \left( 1 - \frac{1}{e} - \epsilon_1 \right) OPT - \sum_{x=1}^{k} \left( 1 - \frac{1}{k} \right)^{k - x} \left[ F_{t_x} \left( S_{j, x} \right) - F_{t_{x+1}} \left( S_{j, x} \right) \right]
\end{align*}
\end{small}
\end{proof}

To check how close the above approximation guarantee is to $1 - 1/e - \epsilon_1$, the guarantee for the SOTA on IM \cite{TXS14,TSX15}, we ran simulations on $5$ networks (\S\ref{sec:exp}) and found the maximum difference $\gamma$ between the two values across all networks to be $4.7 \times 10^{-14}$ (incremental) and 0.124 (fully dynamic), showing that the algorithm approximates very well in practice.

Using Lemma \ref{lem:greedy}, Theorem 1 of \cite{kalimeris2019robust} and Lemma 3.13 of \cite{peng2021dynamic}, we can derive the following quality guarantee for Alg. \ref{alg:find_seeds}.

\begin{lemma}
    Let $\gamma$ be as defined in Lemma \ref{lem:greedy}. In every time step $t$, Alg. 2 returns the uniform distribution $\mathcal{U}$ over solutions $S_t^{(1)}, \ldots, S_t^{(T)}$ s.t., for $T \geq \frac{2 \log l}{\epsilon_2^2}$, with probability at least $1 - \delta_2$,
    \begin{small}
    \begin{equation*}
        \min_{i \in [l]} \mathbb{E}_{S \sim \mathcal{U}} f_{cv, i, t} \left( S \right) \geq \left( 1 - 1/e - \gamma -\epsilon_1 \right) \max_{S : \left| S \right| \leq k} \min_{i \in [l]} f_{cv, i, t} ( S) - \epsilon_2
    \end{equation*}
    \end{small}
\end{lemma}

The result of the above lemma implies the following bi-criteria approximation guarantee by returning $S_{tc} = \cup_{i=1}^T S_t^{(i)}$ instead of the uniform distribution over $\left\{ S_t^{(1)}, \ldots, S_t^{(T)} \right\}$.

\begin{corollary}
\label{cor:union}
    The seed set $S_{tc} = \cup_{i=1}^T S_t^{(i)}$ has size at most $Tk$ and, with probability at least $1 - \delta_2$, satisfies
    \begin{small}
    \begin{equation*}
        \min_{i \in [l]} f_{cv, i, t} \left( S_{tc} \right) \geq \left( 1 - 1/e - \gamma -\epsilon_1 \right) \max_{S_t : \left| S_t \right| \leq k} \min_{i \in [l]} f_{cv, i, t} (S) - \epsilon_2
    \end{equation*}
    \end{small}
\end{corollary}

The above result provides a guarantee on the minimum RR set coverage $f_{cv,i,t}(\cdot)$ across all values of $i$. Using that, we now provide a guarantee on the minimum true influence spread $\sigma_{\theta_i,t}(\cdot)$ across all sampled hyperparameter values.

\begin{lemma}
\label{lem:f2sigma}
    Let $\gamma$ be as defined in Lemma \ref{lem:greedy}. With probability at least $1 - \delta_1 / 2 - \delta_2$,
    \begin{small}
    \begin{equation*}
        \min_{i \in [l]} \sigma_{\theta_i, t} (S_{tc}) \geq \left( 1 - 1/e - \gamma - 3\epsilon_1 \right) \max_{S \subseteq V_t : |S| \leq k} \min_{i \in [l]} \sigma_{\theta_i, t} (S) - \epsilon_2
    \end{equation*}
    \end{small}
\end{lemma}
\begin{proof}
    For ease of notation, we shall replace $\sigma_{\theta_i,t}(\cdot)$ with $\sigma_{i,t}(\cdot)$ in this proof. Let $AVG_i \cdot m_0$ denote the expected number of edges checked while generating a random RR set with hyperparameter value $\theta_i$. Similar to Lemma 3.3 of \cite{peng2021dynamic}, it can be shown that, with probability at least $1 - \frac{2\delta}{n_0^{k/8}}$,
    \begin{equation*}
        K_i \in \left[ (1 - \epsilon) \frac{R}{AVG_i}, (1 + \epsilon) \frac{R}{AVG_i} \right]
    \end{equation*}
    Assume that the above holds. Note that $f_{cv, i, t}(S) = \frac{n_0}{K_i} \sum_{v \in V_t} x_{i, t, v, S}$ and $\mathbb{E}[f_{cv, i, t}(S)] = \sigma_{i, t}(S)$ for all $S$ and $i$. Consider any seed set $S_t$ of size at most $Tk$. Let $i_{S_t} \in \left\{ \arg\min_{i \in [l]} \sigma_{i, t}(S), \arg\min_{i \in [l]} f_{cv, i, t}(S) \right\}$, $OPT_t = \max_{S \subseteq V_t : |S| \leq k} \min_{i \in [l]} \sigma_{i, t}(S)$ and write $\sigma_{i_{S_t}, t}(S_t) = \lambda OPT_t$. Using the Chernoff bound, we have
    \begin{align*}
        &\Pr \left( f_{cv, i_{S_t}, t}(S_t) - \sigma_{i_{S_t}, t}(S_t) > \epsilon_1 OPT_t \right) \\
        &= \Pr \left( f_{cv, i_{S_t}, t}(S_t) > \left( 1 + \frac{\epsilon_1}{\lambda} \right) \sigma_{i_{S_t}, t}(S_t) \right) \\
        &\leq \exp \left( - \frac{\frac{\epsilon_1^2}{\lambda^2}}{2 + \frac{\epsilon_1}{\lambda}} \cdot \frac{K_{i_{S_t}}}{n_0} \cdot \sigma_{i_{S_t}, t}(S_t) \right)
    \end{align*}
    When $\lambda > \epsilon_1$, the exponent obeys
    \begin{align*}
        &\frac{\frac{\epsilon_1^2}{\lambda^2}}{2 + \frac{\epsilon_1}{\lambda}} \cdot \frac{K_{i_{S_t}}}{n_0} \cdot \sigma_{i_{S_t}, t}(S_t) \geq \frac{\epsilon_1^2}{3\lambda^2} \cdot \frac{K_{i_{S_t}}}{n_0} \cdot \lambda \cdot OPT_t \\
        &\geq \frac{\epsilon_1^2 K_{i_{S_t}}}{3n_0} \cdot \frac{k^2}{2n_0} \geq \frac{\epsilon_1^2 k^2}{6n_0^2} \cdot (1 - \epsilon_1) \frac{R}{AVG_i} \geq \frac{k^3}{12n_0^2} \log \left( \frac{n_0}{\delta} \right)
    \end{align*}
    When $\lambda \leq \epsilon_1$, the exponent obeys
    \begin{align*}
        &\frac{\frac{\epsilon_1^2}{\lambda^2}}{2 + \frac{\epsilon_1}{\lambda}} \cdot \frac{K_{i_{S_t}}}{n_0} \cdot \sigma_{i_{S_t}, t}(S_t) \geq \frac{\epsilon_1}{3\lambda} \cdot \frac{K_{i_{S_t}}}{n_0} \cdot \lambda \cdot OPT_t \\
        &\geq \frac{\epsilon_1^2 k}{3n_0} \cdot K_{i_{S_t}} \geq \frac{\epsilon_1^2 k}{3n_0} \cdot (1 - \epsilon_1) \frac{R}{AVG_i} \geq \frac{k^2}{6n_0} \log \left( \frac{n_0}{\delta} \right)
    \end{align*}
    Hence
    \begin{equation*}
        \Pr \left( f_{cv, i_{S_t}, t}(S_t) - \sigma_{i_{S_t}, t}(S_t) > \epsilon_1 OPT_t \right) \leq \left( \frac{\delta}{n_0} \right) ^{\frac{k^3}{12n_0^2}}
    \end{equation*}
    Similarly, using the Chernoff bound, we have
    \begin{align*}
        &\Pr \left( f_{cv, i_{S_t}, t}(S_t) - \sigma_{i_{S_t}, t}(S_t) < -\epsilon_1 OPT_t \right) \\
        &= \Pr \left( f_{cv, i_{S_t}, t}(S_t) < \left( 1 - \frac{\epsilon_1}{\lambda} \right) \sigma_{i_{S_t}, t}(S_t) \right) \\
        &\leq \exp \left( - \frac{\epsilon_1^2}{2\lambda^2} \cdot \frac{K_{i_{S_t}}}{n_0} \cdot \sigma_{i_{S_t}, t}(S_t) \right) \leq \left( \frac{\delta}{n_0} \right) ^{\frac{k^3}{8n_0^2}}
    \end{align*}
    To summarize,
    \begin{small}
    \begin{gather*}
        \Pr \left( \left| f_{cv, i_{S_t}, t}(S_t) - \sigma_{i_{S_t}, t}(S_t) \right| > \epsilon_1 \max_{S \subseteq V_t : |S| \leq k} \min_{i \in [l]} \sigma_{i, t}(S) \right) \\
        \leq 4 \left( \frac{\delta}{n_0} \right) ^{\frac{k^3}{12n_0^2}}
    \end{gather*}
    \end{small}
    By a union bound over all sets of size at most $Tk$ and both values of $i_{S_t}$,
    \begin{gather*}
        \Pr \left( \forall S_t : |S_t| \leq Tk, \forall i_{S_t}, \left| f_{cv, i_{S_t}, t}(S_t) - \sigma_{i_{S_t}, t}(S_t) \right| \right. \\
        \left. \leq \epsilon_1 \max_{S \subseteq V_t : |S| \leq k} \min_{i \in [l]} \sigma_{i, t}(S) \right) \\
        \geq 1 - 4 \left( \frac{\delta}{n_0} \right) ^{\frac{k^3}{12n_0^2}} \cdot 2 (2n_0)^{Tk} = 1 - \frac{\delta_1}{2}
    \end{gather*}
    By a union bound over the above event and that in Corollary \ref{cor:union}, with probability at least $1 - \frac{\delta_1}{2} - \delta_2$, we have
    \begin{align*}
        &\min_{i \in [l]} \sigma_{i, t} (S_{tc}) = \sigma_{i_{S_{tc}}, t} (S_{tc}) \\
        &\geq f_{cv, i_{S_{tc}}, t} (S_{tc}) - \epsilon_1 \max_{S \subseteq V_t : |S| \leq k} \min_{i \in [l]} \sigma_{i, t}(S) \\
        &\geq \min_{i \in [l]} f_{cv, i, t} (S_{tc}) - \epsilon_1 \max_{S \subseteq V_t : |S| \leq k} \min_{i \in [l]} \sigma_{i, t}(S) \\
        &\geq \left( 1 - 1/e - \gamma - 2\epsilon_1 \right) \max_{S \subseteq V_t : \left| S \right| \leq k} \min_{i \in [l]} f_{cv, i, t}(S) - \epsilon_2 \\
        &\geq \left( 1 - 1/e - \gamma - 3\epsilon_1 \right) \max_{S \subseteq V_t : \left| S \right| \leq k} \min_{i \in [l]} \sigma_{i, t}(S) - \epsilon_2
    \end{align*}
    The penultimate step above follows from the event in Corollary \ref{cor:union}. The last step holds since, defining $S_\sigma = \arg\max_{S \subseteq V_t : \left| S \right| \leq k} \min_{i \in [l]} \sigma_{i, t}(S)$,
    \begin{align*}
        & \max_{S \subseteq V_t : \left| S \right| \leq k} \min_{i \in [l]} f_{cv, i, t}(S) \geq \min_{i \in [l]} f_{cv, i, t}(S_\sigma) = f_{cv, i^*, t}(S_\sigma) \\
        &\geq \sigma_{i^*, t}(S_\sigma) - \epsilon_1 \max_{S \subseteq V_t : \left| S \right| \leq k} \min_{i \in [l]} \sigma_{i, t}(S) \\
        &\geq \min_{i \in [l]} \sigma_{i, t}(S_\sigma) - \epsilon_1 \max_{S \subseteq V_t : \left| S \right| \leq k} \min_{i \in [l]} \sigma_{i, t}(S) \\
        &= (1 - \epsilon_1) \max_{S \subseteq V_t : \left| S \right| \leq k} \min_{i \in [l]} \sigma_{i, t}(S)
    \end{align*}
\end{proof}

Using Lemma~\ref{lem:f2sigma}, we can now provide a guarantee for the minimum influence spread across the entire hyperparameter space $\Theta$. This, along with the running time, forms the overall result of this section.

\begin{theorem}
\label{th:i2theta}
    Let $\gamma$ be as defined in Lemma \ref{lem:greedy}. Let $c_p$ denote the maximum number of phases during which the size of the graph remains within a factor of 2 of the original. Within each such phase, let $c_r$ denote the number of stages within which the number of edges scanned to update $\mathcal{H}_{est,i}$ for each $i \in [l]$ remains below the threshold after which it restarts from scratch. Define $S_{tc} = \cup_{i=1}^T S_t^{(i)}$. With probability at least $1 - \delta_1 - \delta_2$, in every time step $t$,
    \begin{small}
    \begin{equation*}
        \min_{\theta \in \Theta} \sigma_{\theta, t} (S_{tc}) \geq \left( 1 - 1/e - \gamma - 3\epsilon_1 \right) \max_{S \subseteq V_t : |S| \leq k} \min_{\theta \in \Theta} \sigma_{\theta, t} (S) - 2\epsilon_2
    \end{equation*}
    \end{small}
    for $l \geq d \left( \frac{2Bdmn}{\epsilon_2} \right)^d \log \left( \frac{2Bdmn}{\epsilon_2\delta_2} \right)$ and $T \geq \frac{2 \log l}{\epsilon_2^{2}}$, with amortized running time $O \left( \frac{n^2 l^2 T c_p c_r}{k^2 \epsilon_1^3} \cdot \log n \cdot \log \left( \frac{n^{Tk}}{\delta_1} \right) \right)$.
\end{theorem}

\begin{proof}
From Lemma \ref{lem:f2sigma} and from Lemma 5 of \cite{kalimeris2019robust}, the quality guarantee of Theorem \ref{th:i2theta} holds with probability at least $1 - \frac{\delta_1}{2} - \delta_2$.

Regarding the running time, note that the size of the network can be doubled for at most $c_p$ times, i.e. there are at most $c_p$ phases where the size of the network remains within a factor of two of that at the beginning of the phase. Within each of the above phases, our algorithm restarts for at most $c_r$ times, and on every such restart, the value of $p_i$ for each $i \in [l]$ is estimated once using Lines 3-9 of Algorithm 1. This, as also the construction of $\mathcal{H}_{est,i}$ and $\mathcal{H}_{cv,i}$, each has $O(k \epsilon_1^{-2} \log(n/\delta))$ amortized time per network update with probability at least $1 - \frac{\delta_1}{2}$; this follows from Lemmas 3.4, 3.6 and 3.11 of \cite{peng2021dynamic}. Note that once the $\mathcal{H}_{cv,i}$s are updated, the seed set also needs to be updated using Algorithm 2. Analogous to Lemma 3.14 of \cite{peng2021dynamic}, it can be shown that the amortized running time of Algorithm 3 is $O(\epsilon_1^{-1} \log n)$, which means the one for Algorithm 2 is $O \left( \frac{lT}{\epsilon_1} \log n \right)$. Thus, the overall running time is
\begin{align*}
    & O \left( l c_p c_r \left( \frac{2k}{\epsilon_1^2} \log \left( \frac{n}{\delta} \right) + \frac{k}{\epsilon_1^2} \log(n/\delta) \cdot \frac{lT}{\epsilon_1} \log n \right) \right) \\
    & \leq O \left( \frac{n^2 l^2 T c_p c_r}{k^2 \epsilon_1^3} \cdot \log n \cdot \log \left( \frac{n^{Tk}}{\delta_1} \right) \right)
\end{align*}

From the union bound, both the quality guarantee and the running time hold with probability at least $1 - \delta_1 - \delta_2$.
\end{proof}

\spara{Remarks:} 
{\bf (1)} Th.~\ref{th:i2theta} provides minimum values of $l$ and $T$ needed for a desired solution quality, with higher values leading to better solutions. However, in our experiments (see \S~\ref{sec:sensitivity}), we find that we can obtain reasonably good results even with much smaller values of $l$ and $T$, and the solution quality converges (remains nearly the same) beyond those values. So in practice we can run \algoname\ for much smaller $l$ and $T$ (and faster), while achieving a reasonably good solution. {\bf (2)} 
For the incremental setting, 
the values of $c_p$ and $c_r$ in Th.~\ref{th:i2theta} are at most $O(\log n)$ \cite{peng2021dynamic}. Also, according to Line \ref{line:cond_k} of Alg. \ref{alg:greedy_insert}, once all threads' seed sets have $k$ nodes each, they will not be further updated (since they can only increase in size), which means we can also avoid bookkeeping computations like updating the marginal gain of every node for subsequent seed additions. However, this cannot be done for the fully dynamic setting because, in that case, even if a thread's seed set has $k$ nodes, an edge deletion will require replacement of one seed with another, which requires the marginal gains of each node to be up to date (Alg. \ref{alg:greedy_remove}). These factors account for the running time being longer for the fully dynamic setting than for the incremental one, as proven empirically in \S\ref{sec:exp}.

%% file: experiments.tex
\section{Experiments}
\label{sec:exp}
We evaluate next \algoname's effectiveness and efficiency. It was implemented in C++, run on a Linux server with 2 GHz AMD CPU, 512GB RAM. 

\begin{figure*}
    \centering
    \begin{subfigure}{0.195\textwidth}
    \centering
    \includegraphics[scale=0.21]{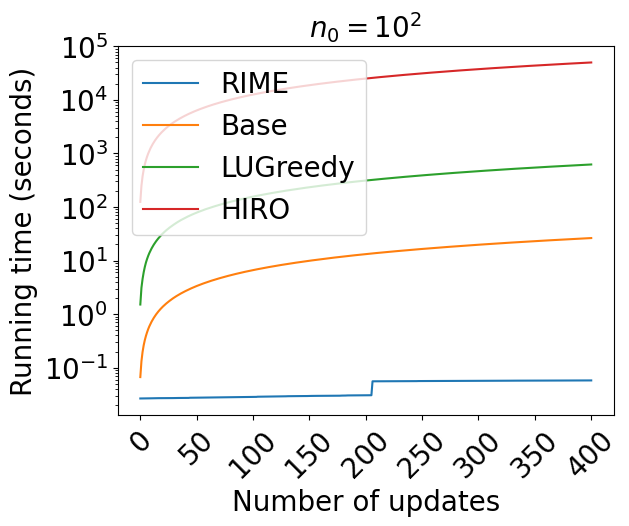}
    \end{subfigure}
    \begin{subfigure}{0.195\textwidth}
    \centering
    \includegraphics[scale=0.21]{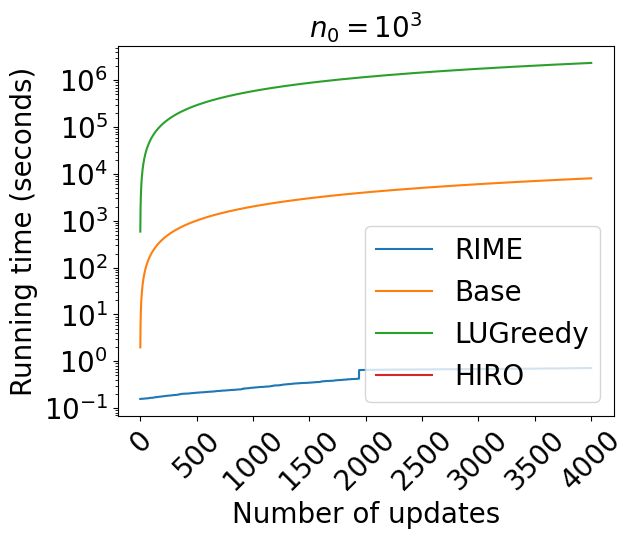}
    \end{subfigure}
    \begin{subfigure}{0.195\textwidth}
    \centering
    \includegraphics[scale=0.21]{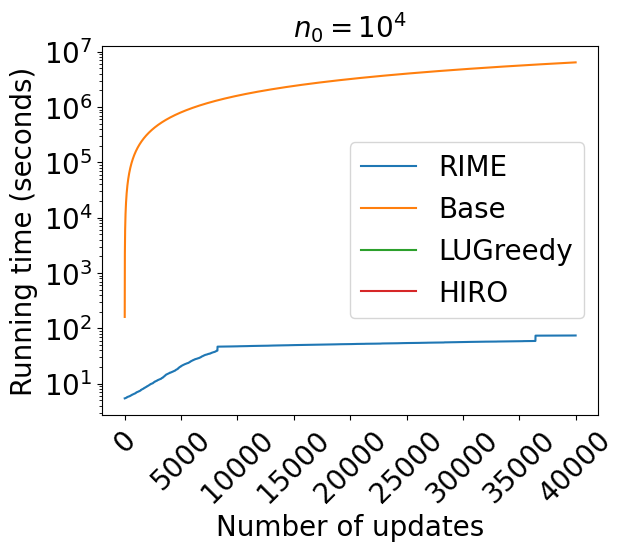}
    \end{subfigure}
    \begin{subfigure}{0.195\textwidth}
    \centering
    \includegraphics[scale=0.21]{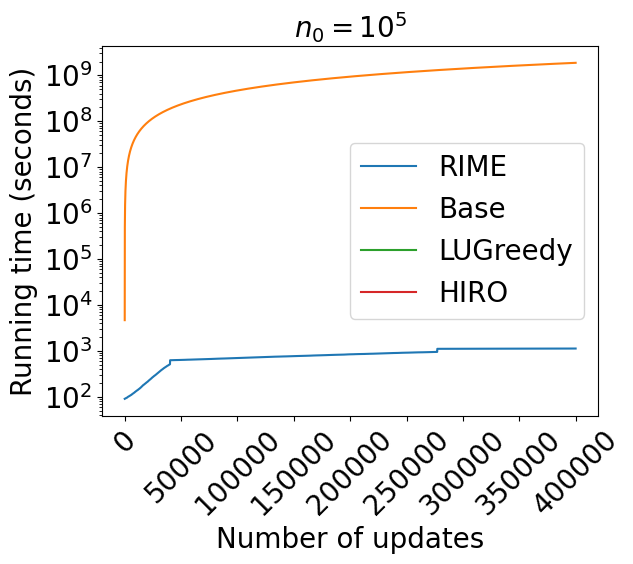}
    \end{subfigure}
    \begin{subfigure}{0.195\textwidth}
    \centering
    \includegraphics[scale=0.21]{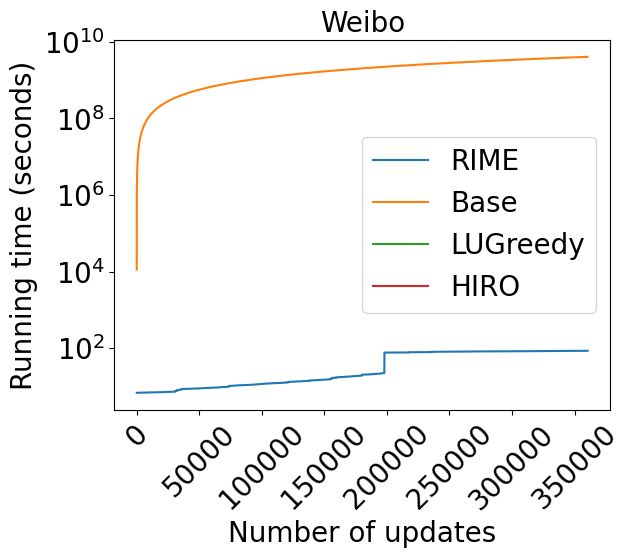}
    \end{subfigure}
    \begin{subfigure}{0.195\textwidth}
    \centering
    \includegraphics[scale=0.21]{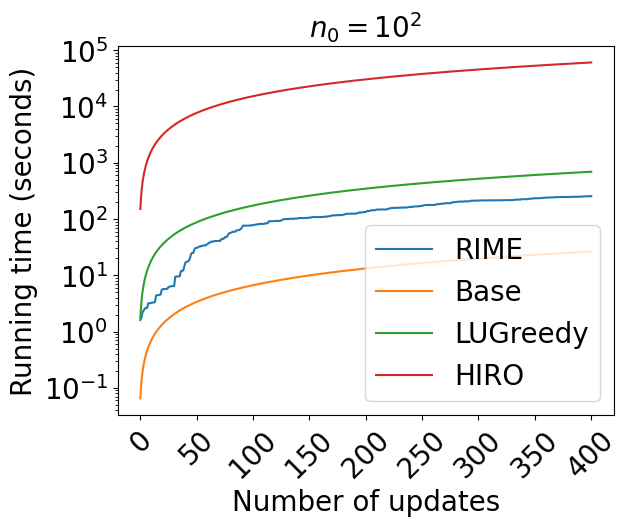}
    \end{subfigure}
    \begin{subfigure}{0.195\textwidth}
    \centering
    \includegraphics[scale=0.21]{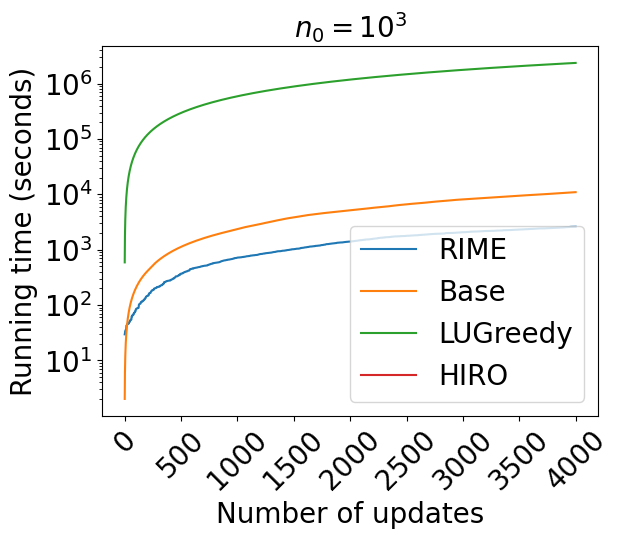}
    \end{subfigure}
    \begin{subfigure}{0.195\textwidth}
    \centering
    \includegraphics[scale=0.21]{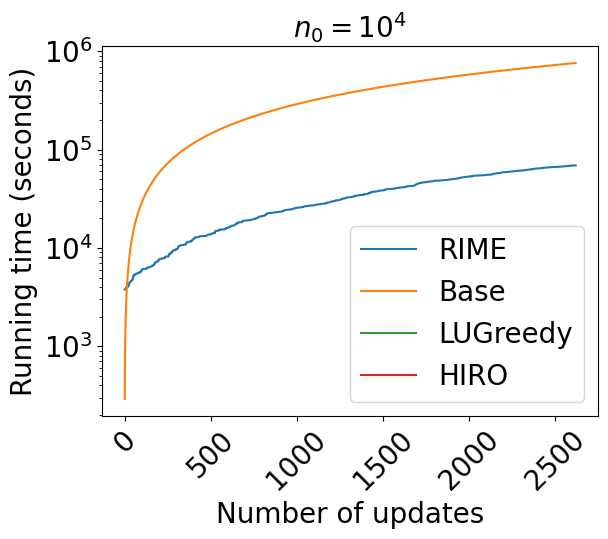}
    \end{subfigure}
    \begin{subfigure}{0.195\textwidth}
    \centering
    \includegraphics[scale=0.21]{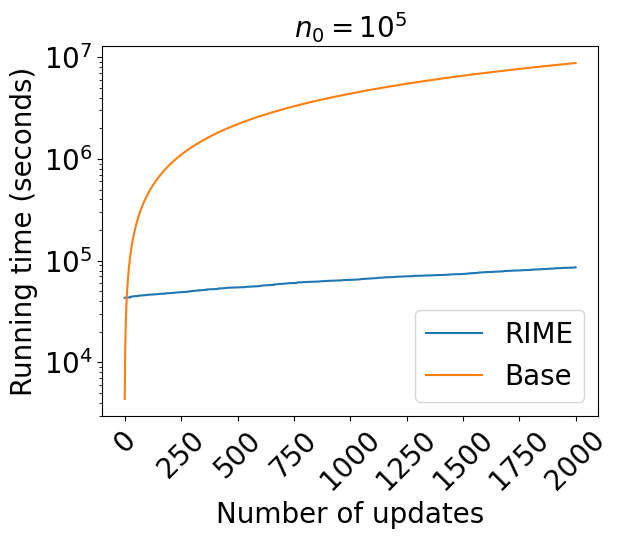}
    \end{subfigure}
    \begin{subfigure}{0.195\textwidth}
    \centering
    \includegraphics[scale=0.21]{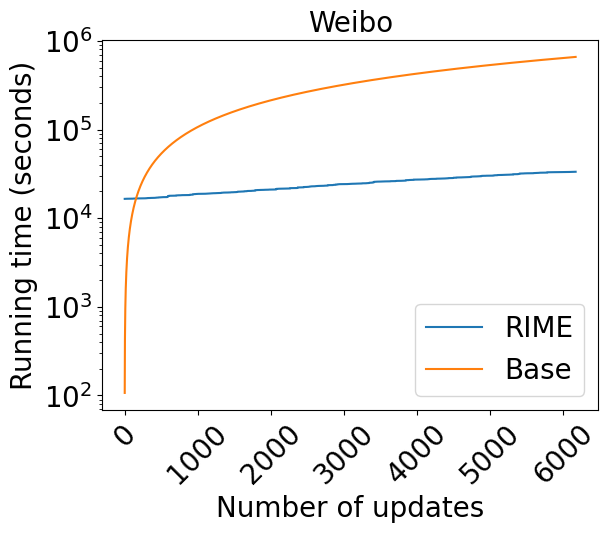}
    \end{subfigure}
    \caption{\small{Running time comparison on synthetic / real-world networks for incremental (top row) / fully dynamic (bottom row) settings.}}
    \label{fig:time}
\end{figure*}

\subsection{Experimental Setup}
\label{sec:setup}
\spara{Networks.} We generate synthetic random networks with $n_0 = 10^x$, $x \in \{2, 3, 4, 5\}$, and $m_0 = 2n_0$. We also test our methods on the real-world Weibo social network \cite{zhang2013social}. From the original dataset, we extract the largest connected component of the subgraph containing edges with at least 30 reposts between the corresponding users. The resultant network has $n_0 = 99523$ and $m_0 = 180216$. In the incremental setting, we also test our methods on the social network LiveJournal \cite{backstrom2006group} ($n_0 = 4847571, m_0 = 68993773$) and two more synthetic networks ($n_0 \in \{10^6, 10^7\}, m_0 = 2n_0$).

\spara{Hyperparametric model.} 
In each of our networks, we use the sigmoid function as the hyperparametric model to determine the edge probabilities, i.e., $H(\theta, x_e) = 1 / [1 + \exp{(-\theta^T x_e)}]$ as in \cite{kalimeris2019robust}.

\spara{Features.} For our synthetic networks, following \cite{kalimeris2019robust}, we generate 3 random features for each node, for a 6-dimensional hyperparameter for each edge. As in \S\ref{sec:robust}, the hyperparameter space is a $d$-cube centred at the origin. In  Weibo, we use as features user attributes (followers / friends / status counts, gender, verified status). The categorical features are binarized. \eat{We assume that} The resultant hyperparameter space, with dimension $d = 16$, is a $d$-cube centred at the value for which the likelihood of the repost cascade data (available with the dataset) is maximized \cite{kalimeris2018learning}. This makes the space of possible edge probabilities more realistic than that with the hypercube centred at the origin, while not affecting our algorithms in any way.

\spara{Network updates.} For each of our networks initially with $n_0$ nodes and $m_0$ edges, we generate $2m_0$ random updates of all possible kinds (node / edge insertion / removal) for both the incremental setting and the fully dynamic setting.

\spara{Methods compared.} For both the incremental setting and the fully dynamic one, we compare \algoname\ with baselines that run the following methods from scratch after each network update:\footnote{Another option is to run a baseline after a batch of updates instead of after every update. However, that results in the baseline's solution quality getting much  poorer, e.g., from $165.574$ to $115$ with a batch of size $100$ on the graph with $n_0 = 10^3$.}
(i) HIRO (resp. (ii) BASE), which runs the method in \cite{kalimeris2019robust} by generating RR sets from $\theta$ randomly chosen nodes (resp. once from every node) for each MWU iteration; and 
(iii) LUGreedy \cite{chen2016robust}, which assumes confidence intervals for the edge probabilities, runs the greedy algorithm twice over both the lower and the upper boundaries of the intervals, and returns the better of the two above solutions assuming the lower boundaries as true probabilities.

\spara{Metrics.} We compare \algoname's efficiency against the baselines on \emph{total running time} (across all network updates). We compare \algoname's quality by computing, after all updates, the minimum (approximate) influence spread of the returned seed sets across the sampled hyperparameter values.\footnote{Each value is an average over 10 algorithm runs, generally with a small std. dev.}

\spara{Parameters.} We vary the \textit{number of seeds} $k \in \{10, 25, 50\}$ \cite{kalimeris2019robust}\eat{, with $k = 50$ being the default value}. We vary the \textit{radius of the hyperparameter space} $B \in \{0, 0.25, 0.5, 0.75, 1\}$ and study its effect on the values of $l$ and $T$ in practice (the theoretical dependence is shown in Theorem \ref{th:i2theta}). We set $B = 1$ by default.

\subsection{Results}
\label{sec:results}

\begin{figure}[tb!]
    \centering
    \begin{subfigure}[t]{0.22\textwidth}
    \centering
    \includegraphics[scale=0.26]{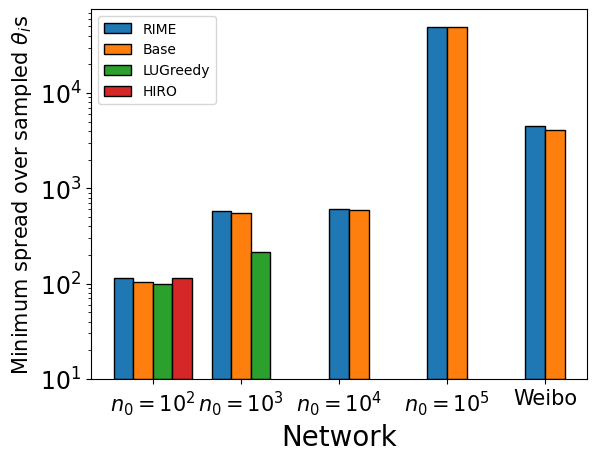}
    \end{subfigure}
    \begin{subfigure}[t]{0.22\textwidth}
    \centering
    \includegraphics[scale=0.26]{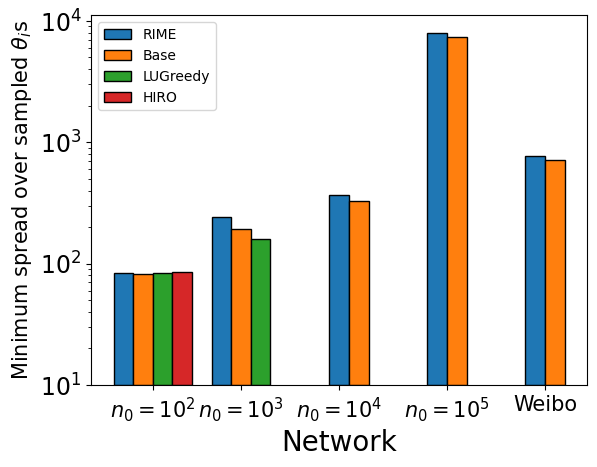}
    \end{subfigure}
    \caption{\small{Solution quality comparison on synthetic / real-world networks for incremental (left) and fully dynamic (right) settings.}}
    \label{fig:score}
\end{figure}

\spara{Efficiency.} We compare the running time of \algoname\ against the baselines in Fig. \ref{fig:time} and Figs. \ref{fig:more}\cref{fig:more_time1,fig:more_time2,fig:more_time3}. A missing line means that the corresponding method did not finish in a day while running on the very first graph without any updates. It is clear that, in most cases, \algoname\ is several orders of magnitude faster than the competitors. Also, for the incremental setting \algoname\ is 2-3 orders of magnitude faster than that for the fully dynamic setting. This is explained by the fact that we can avoid several bookkeeping computations for the incremental setting (see the last paragraph of \S\ref{sec:analysis}). Observe the jumps in the plots for \algoname. They correspond to \algoname\ invoking the algorithm restarting from scratch (Alg. 1). Notice that in the incremental setting, the slope of each line for \algoname\eat{, i.e., the time required by each method corresponding to a network update,} decreases after each restart. This can be explained as follows. Later restarts take place on bigger graphs than earlier ones. Thus, after later restarts, when an edge $(u, v)$ is added, RR sets containing $v$ are more likely to contain $u$ as well, compared to earlier restarts, since there are more ways of reaching $u$ on a bigger graph than on a smaller one. Note that only those RR sets containing $v$ but not $u$ need to be augmented (\S\ref{sec:algorithm}). This leads to fewer RR set augmentations after later restarts, thereby lowering the running time. Finally, we would like to point out that \algoname\ scales much better than our baselines, at least up to graphs with $10^5$ nodes.\footnote{The problem we study in this paper has several well-documented sources of hardness. Therefore, while \algoname\ may not scale to large, Twitter-like diffusion networks, it remains reasonably applicable for networks with node counts of orders of tens-of-thousands to hundreds-of-thousands. However, we believe this scale is entirely relevant for many application scenarios, including the ones we mention in \S\ref{sec:intro}.}

\begin{figure}[tb!]
    \centering
    \begin{subfigure}[t]{0.22\textwidth}
    \centering
    \includegraphics[scale=0.25]{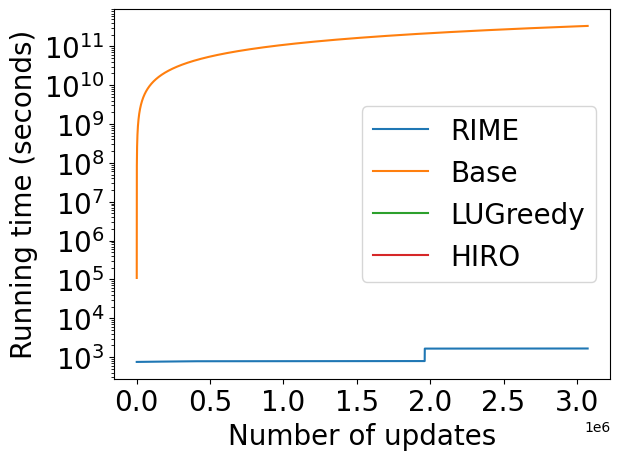}
    \caption{\small{Running time, $n_0 = 10^6$}}
    \label{fig:more_time1}
    \end{subfigure}
    \begin{subfigure}[t]{0.22\textwidth}
    \centering
    \includegraphics[scale=0.25]{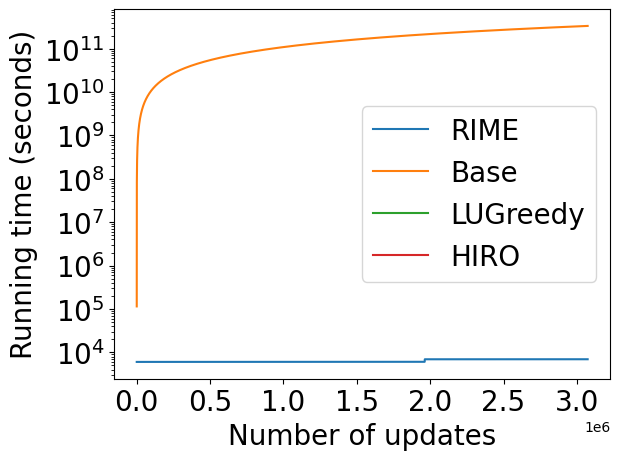}
    \caption{\small{Running time, $n_0 = 10^7$}}
    \label{fig:more_time2}
    \end{subfigure}
    \begin{subfigure}[t]{0.22\textwidth}
    \centering
    \includegraphics[scale=0.25]{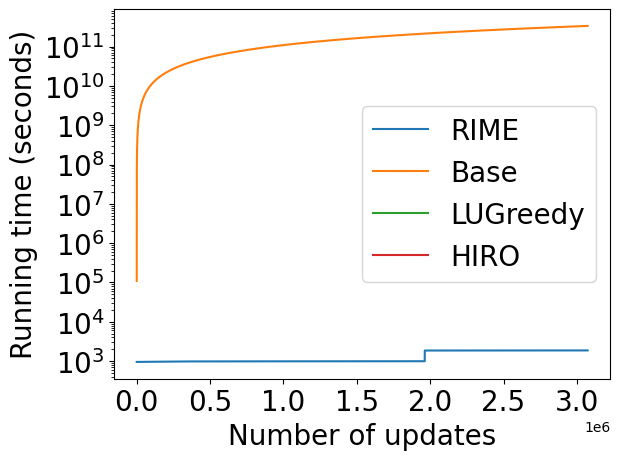}
    \caption{\small{Running time, LiveJournal}}
    \label{fig:more_time3}
    \end{subfigure}
    \begin{subfigure}[t]{0.22\textwidth}
    \centering
    \includegraphics[scale=0.25]{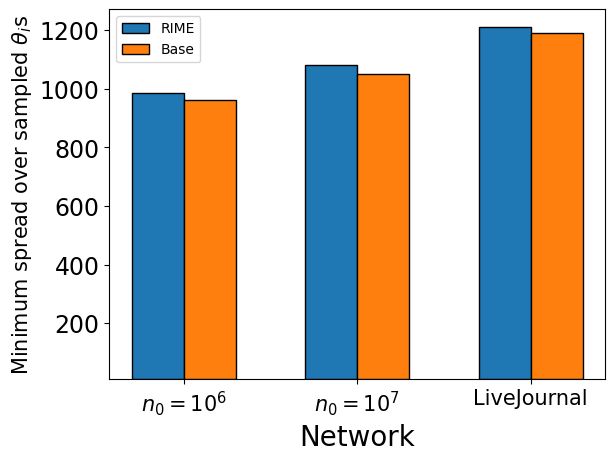}
    \caption{\small{Solution quality}}
    \label{fig:more_quality}
    \end{subfigure}
    \caption{\small{Running time and solution quality comparison on more networks in the incremental setting.}}
    \label{fig:more}
\end{figure}

\spara{Effectiveness.} The solution quality comparisons for both settings are shown in Figs. \ref{fig:score} and \ref{fig:more_quality}. A missing bar means that the corresponding method did not finish in a day while running on the initial graph without any updates. Clearly, the best baseline is HIRO, which does not scale to graphs with $10^3$ nodes or more. Even on the graphs on which HIRO does complete, \algoname's quality is comparable to HIRO's. LUGreedy returns lower-quality solutions than \algoname, while also not scaling to graphs with $10^4$ nodes or more. BASE finishes on all our graphs, but returns slightly lower-quality solutions than \algoname, while being orders of magnitude slower.

\eat{
To summarize, \algoname\ is orders of magnitude faster than the baselines (\emph{efficient}), while returning solutions of comparable quality (\emph{effective}). 
}

\begin{figure}[tb!]
    \centering
    \begin{subfigure}[t]{0.22\textwidth}
    \centering
    \includegraphics[scale=0.26]{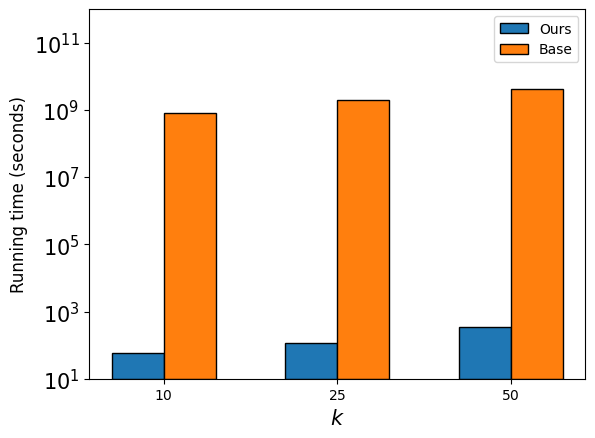}
    \caption{\small{Running time}}
    \end{subfigure}
    \begin{subfigure}[t]{0.22\textwidth}
    \centering
    \includegraphics[scale=0.26]{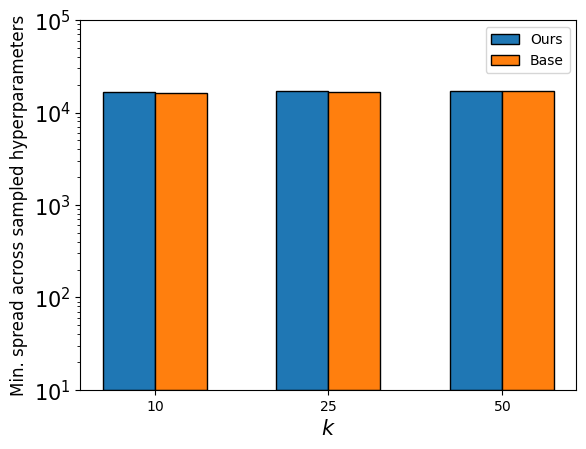}
    \caption{\small{Solution quality}}
    \end{subfigure}
    \caption{\small{Variation of running time and solution quality with $k$ on the Weibo network in the incremental setting.}}
    \label{fig:varyk}
\end{figure}

\begin{figure}[tb!]
    \centering
    \begin{subfigure}[t]{0.22\textwidth}
    \centering
    \includegraphics[scale=0.26]{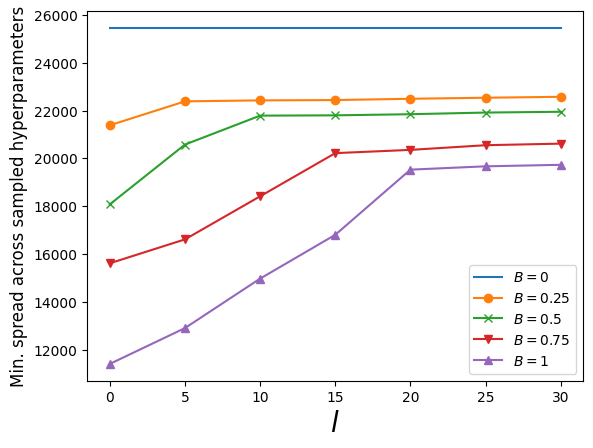}
    \caption{\small{Impact on $l$}}
    \end{subfigure}
    \begin{subfigure}[t]{0.22\textwidth}
    \centering
    \includegraphics[scale=0.26]{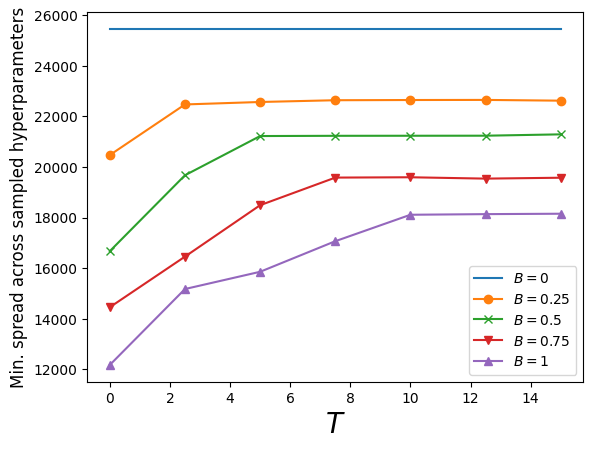}
    \caption{\small{Impact on $T$}}
    \end{subfigure}
    \caption{\small{Variation of solution quality with $l$ and $T$ for different values of $B$ on the Weibo network in the incremental setting.}}
    \label{fig:varyb}
\end{figure}

\subsection{Parameter Sensitivity Analysis}
\label{sec:sensitivity}

\spara{Impact of $k$.} Figure \ref{fig:varyk} shows the effect of $k$ on the total running time and the quality of our solution after all updates in Weibo. Clearly, for every value of $k$, our method is orders of magnitude faster than BASE, while returning a solution of similar quality. As before, LUGreedy and HIRO are missing for each $k$ because they did not finish in a day while running on the very first graph without any updates.

\spara{Impact of $B$ on $l$ and $T$.} For different values of $B$, we study the variation of the solution quality with $l$ and $T$. Figure \ref{fig:varyb} shows that the quality converges after very small values of $l$ and $T$, even though theoretically higher values should lead to better solutions. Notice that the quality reduces with $B$, since the minimum influence spread over a larger hyperparameter space is lower than that over a smaller one. Also, we find that there is a relation between the converging values of $l$ and $T$ for different values of $B$ (though not exactly the same as shown in Theorem \ref{th:i2theta}). For instance, with $B = 1$, we find that convergence is achieved at $l = 20$ and $T = 10$, while with $B = 0.5$, the same is achieved at $l = 10$ and $T = 5$. Thus, for any dataset, once we find the converging values of $l$ and $T$ for $B = 1$ (which corresponds to the highest level of uncertainty about the edge probabilities), we can infer the same for other values of $B$.

%% file: conclusion.tex
\section{Conclusion}
We studied the problem of robust influence maximization over dynamic diffusion networks, where the diffusion probabilities are a function of the node features and a global hyperparameter, and the network evolves by node / edge insertions / deletions. 
\eat{
In the incremental setting,  when nodes and edges are only added to the network,}
We proposed an approximation solution using multiplicative weight updates and a greedy algorithm, with theoretical quality guarantees. 
\eat{
For the fully dynamic setting, i.e., when both node / edge insertions and deletions can occur, we propose a similar, albeit heuristic solution. 
}
We empirically assessed our method's performance on synthetic \& real networks and found that it is generally orders of magnitude faster than the baselines, while returning a seed set of comparable quality. A potential future direction is extending the \algoname\ framework to other diffusion models (e.g. Linear Threshold and triggering models) and other forms of network dynamicity. 